\documentclass[a4paper]{amsart}

\usepackage[utf8]{inputenc}  
\usepackage{amsmath}
\usepackage{thmtools} 
\usepackage{amssymb}
\usepackage[dvipsnames]{xcolor}
\usepackage[hypertexnames=false,linkcolor=black, colorlinks=true, citecolor=black, filecolor=black,urlcolor=black]{hyperref}
\usepackage{mathtools}
\usepackage{enumitem} 
\usepackage[normalem]{ulem} 
\usepackage{xifthen}

\usepackage[backref=true, style=alphabetic]{biblatex}
\usepackage[xcolor]{changebar}

\usepackage[autostyle=true]{csquotes}

\usepackage{tikz}
\usetikzlibrary{ducks}

\usepackage{xspace}

\usepackage{diagbox}

\usepackage{booktabs}
\usepackage{array}
\newcolumntype{C}{>{$}c<{$}}
\newcolumntype{L}{>{$}l<{$}}

\numberwithin{equation}{section}

\declaretheoremstyle[
  headfont=\normalfont\bf,
  bodyfont=\normalfont\itshape,
  spaceabove=1em,
  spacebelow=1em,
]{results}

\declaretheorem[
  style=results,
  title=Theorem,
  numberwithin=section,
]{theorem}
\declaretheorem[
  style=results,
  title=Lemma,
  sibling=theorem,
]{lemma}

\declaretheorem[
  style=results,
  title=Corollary,
  sibling=theorem,
]{corollary}

\declaretheoremstyle[
  headfont=\normalfont\bf,
  bodyfont=\normalfont,
  spaceabove=1em,
  spacebelow=1em,
  qed={$\lhd$}
]{definitions}

\declaretheorem[
  style=definitions,
  title=Definition,
  sibling=theorem,
]{definition}

\declaretheoremstyle[
  headfont=\normalfont\bf,
  bodyfont=\normalfont,
  spaceabove=1em,
  spacebelow=1em,
]{other}

\declaretheorem[
  style=other,
  title=Remark,
  sibling=theorem,
]{remark}

\declaretheorem[
  style=other,
  title=Notation,
  sibling=theorem,
]{notation}
\declaretheorem[
  style=other,
  title=Example,
  sibling=theorem,
]{example}

\declaretheorem[
  style=other,
  title=Observation,
  sibling=theorem,
]{observation}

\theoremstyle{remark}
\newtheorem*{acknowledgements}{Acknowledgements}

\DeclareUnicodeCharacter{211D}{$\mathbb{R}$}

\newcommand{\sd}[1]{\ensuremath{\textsc{SizeDepth}_{#1}\xspace}}

\newcommand{\usd}[1]{\ensuremath{\textsc{UnbSizeDepth}_{#1}\xspace}}

\newcommand{\ts}[1]{\ensuremath{\textsc{TimeSpace}_{#1}\xspace}}

\newcommand{\AC}{\ensuremath{\mathrm{AC}}\xspace}
\newcommand{\NC}{\ensuremath{\mathrm{NC}}\xspace}

\newcommand{\FO}{\ensuremath{\mathrm{FO}}\xspace}

\newcommand{\ACRing}[2][]{\mathrm{AC}^{#2}_{#1}}
\newcommand{\NCRing}[2][]{\mathrm{NC}^{#2}_{#1}}

\newcommand{\FORing}[2][]{\ifthenelse{\equal{#1}{}}{\mathrm{FO}_{#2}}{\mathrm{FO}_{#2}[#1]}}

\newcommand{\SUM}{\ensuremath{\textnormal{SUM}_{R}}\xspace}
\newcommand{\SUMR}{\ensuremath{\textnormal{SUM}_{\R}}\xspace}
\newcommand{\PROD}{\ensuremath{\textnormal{PROD}_{R}}\xspace}
\newcommand{\PRODR}{\ensuremath{\textnormal{PROD}_{\R}}\xspace}
\newcommand{\MAX}{\ensuremath{\textnormal{MAX}_{R}}\xspace}

\newcommand{\GPR}{\ensuremath{\mathrm{GPR}}\xspace}
\newcommand{\GFR}{\ensuremath{\mathrm{GFR}_{R}}\xspace}
\newcommand{\GFRi}{\ensuremath{\GFR^i}\xspace}
\newcommand{\GFRb}{\ensuremath{\mathrm{GFR}_{R, \mathrm{bound}}}\xspace}
\newcommand{\GFRbi}{\ensuremath{\GFRb^i}\xspace}
\newcommand{\BIT}{\text{BIT}\xspace}

\newcommand{\bO}{\ensuremath{\mathcal{O}}}

\newcommand{\C}{\ensuremath{\mathbb{C}}\xspace}
\newcommand{\Q}{\ensuremath{\mathbb{Q}}\xspace}
\newcommand{\R}{\ensuremath{\mathbb{R}}\xspace}
\newcommand{\N}{\ensuremath{\mathbb{N}}\xspace}
\newcommand{\Z}{\ensuremath{\mathbb{Z}}\xspace}
\newcommand{\F}{\ensuremath{\mathbb{F}}\xspace}

\DeclareMathOperator{\sign}{sign}
\DeclareMathOperator{\val}{val}

\newcommand{\struc}[1]{\ensuremath{\mathcal{#1}}}
\DeclarePairedDelimiter{\abs}{\lvert}{\rvert}

\newcommand{\struct}{\textsc{struc}}
\newcommand{\structR}{\ensuremath{\struct_{R}}}
\newcommand{\ol}[1]{\ensuremath{\overline{#1}}\xspace}

\newcommand{\maxi}[1][i]{\ensuremath{{\textbf{max}_{#1}}}}
\newcommand{\maxib}[1][i]{\ensuremath{{\textbf{max}_{#1, \mathrm{bound}}}}}

\newcommand{\sumi}[1][i]{\ensuremath{{\textbf{sum}_{#1}}}}
\newcommand{\sumib}[1][i]{\ensuremath{{\textbf{sum}_{#1,\mathrm{bound}}}}}

\newcommand{\prodi}[1][i]{\ensuremath{{\textbf{prod}_{#1}}}}
\newcommand{\prodib}[1][i]{\ensuremath{{\textbf{prod}_{#1,\mathrm{bound}}}}}

\newcommand{\cC}{\ensuremath{\mathcal{C}}}

\newcommand{\simsubseteq}{\ensuremath{\subseteq_\mathrm{sim}}}
\newcommand{\simsubsetneq}{\ensuremath{\subsetneq_\mathrm{sim}}}
\newcommand{\simequal}{\ensuremath{=_\mathrm{sim}}}

\newcommand{\Arb}[1]{\ensuremath{\mathrm{Arb}_{#1}}}

\newcommand{\un}{\ensuremath{\mathrm{unary}\xspace}}
\newcommand{\uval}{\ensuremath{\mathrm{uval}\xspace}}
\renewcommand{\d}{\ensuremath{d}\xspace}

\title[Logical Characterizations of Circuit Classes]{Logical Characterizations of algebraic Circuit Classes over Integral domains}

\author[Barlag]{Timon Barlag}
\address{ 
	Institut f\"ur Theoretische Informatik, 
	Leibniz Universit\"at Hannover, 
	30167 Hannover, Germany}
\email{gaube@thi.uni-hannover.de}
\email{barlag@thi.uni-hannover.de}

\author[Chudigiewitsch]{Florian Chudigiewitsch}
\address{
Institut f\"ur Theoretische Informatik,
Universit\"at zu L\"ubeck,
23562 Lübeck, Germany
}
\email{fch@tcs.uni-luebeck.de}

\author[Gaube]{Sabrina A. Gaube}

\date{\today}



\begin{document}

\begin{abstract}

We present an adapted construction of algebraic circuits over the reals introduced by Cucker and Meer 
to arbitrary infinite integral domains and generalize the $\ACRing[\R]{}$ and $\NCRing[\R]{}$-classes for this setting. 
We give a theorem in the style of Immerman's theorem 
which shows that for these adapted formalisms, sets decided by circuits of constant depth and polynomial size are the same as sets definable by a suitable adaptation of first-order logic.
Additionally, we discuss a generalization of the guarded predicative logic by Durand, Haak and Vollmer and we show characterizations for the $\ACRing[R]{}$ and $\NCRing[R]{}$ hierarchy.
Those generalizations apply to the Boolean $\mathrm{AC}$ and $\mathrm{NC}$ hierarchies as well. Furthermore, we introduce a formalism to be able to compare some of the aforementioned complexity classes with different underlying integral domains.
\end{abstract}

\maketitle





\section{Introduction}\label{section:introduction}
Boolean circuits as a computational model are a fundamental concept in the study of theoretical computer science. Mainly in computational complexity, Boolean circuits play a central part in the analysis of parallel algorithms and the corresponding complexity classes, enabling a finer view and providing new proof methods, especially for subpolynomial complexity classes. An obvious generalization of Boolean circuits is that instead of dealing with truth values as inputs -- or the field $\Z_2$ for the algebraically minded -- we consider elements from some other algebraic structure. This approach has its roots in the works of Blum, Shub, and Smale, whose model of the \emph{BSS-machine} is able to describe computations over real numbers. Following this, Blum, Shub, Smale, and Cucker~\cite{ComplexityRealComp} also give a generalization to algebraic circuits over the reals.

\subsection{Our Contribution}
In this article, we provide logical characterizations of circuit complexity classes over integral domains. In particular, we define natural extensions of the classical circuit complexity hierarchies $\ACRing{i}$ and $\NCRing{i}$ over arbitrary integral domains. The resulting classes are denoted as $\ACRing[R]{i}$ and $\NCRing[R]{i}$, respectively. We adapt the framework of metafinite model theory to define various extensions of first order logic, which capture these new complexity classes.

We establish a Immerman-style theorem stating that $\FORing{R} = \ACRing[R]{0}$ and provide a framework to establish complexity-theoretic comparisons of classes with different underlying integral domains and give examples over which integral domain the $\ACRing{0}$-classes are equal.

We adapt the $\GPR$-operator, which Durand, Haak and Vollmer use to provide logical characterizations of $\ACRing{1}$ and $\NCRing{1}$~\cite{10.1145/3209108.3209179} to logics over metafinite structures and extend it to be able to characterize the whole $\ACRing[R]{}$ and $\NCRing[R]{}$-hierarchies.

Finally, we define a formalism suitable for comparing complexity classes with different underlying integral domains and we show that a hierarchy of sets of complexity classes emerges, such that each set is able to ``simulate'' the complexity classes from the sets lower in the hierarchy.


\subsection{Related Work}
Another model of computation that is commonly known under the name \enquote{algebraic circuit} are Valiant circuits~\cite{Valiant1979}, which are the foundational model in the emerging subfields of algebraic and geometric complexity theory. They differ from the model we analyze in this work in the way that we use $<$-gates, which are not available in the Valiant setting. This difference is of complexity-theoretical significance since for our model, we have that $\NCRing[\R]{i} \subsetneq \NCRing[\R]{i + 1}$~\cite{DBLP:journals/jc/Cucker92} but we have that $\mathrm{VNC^2} =\mathrm{VNC^3} =\ldots = \mathrm{VP}$~\cite{buergisser} in the Valiant setting for every field, in particular over the reals.

\subsection{Outline of the Paper}
We start with an overview of some central concepts from algebra and circuit complexity, which are needed for the definition of our model. We then establish our model of algebraic circuits for arbitrary integral domains and the analogous complexity classes induced by them in analogy to the Boolean case. Afterwards, we give a logical characterization of the presented circuit classes inspired by classical descriptive complexity. However, since our models now have an infinite component, we build on the foundations of \emph{metafinite model theory}.

We then go on to define first-order logic over $R$ and show that, like in the classical case, we have that $\ACRing[R]{0} = \FORing[\Arb{R}]{R}+\mathrm{SUM}_{R} + \mathrm{PROD}_{R}$.

In Section~\ref{section:gfr}, we give logical characterizations of $\ACRing{i}$ and $\NCRing{i}$. The tool of our choice is an adaptation of the guarded predicative logic of~\cite{10.1145/3209108.3209179}.

In Section~\ref{section:relationship-regarding-rings}, we discuss connections between $\ACRing[R]{0}$-classes over different integral domains.

\begin{acknowledgements}
	We thank Sebastian Berndt and  Anselm Haak for fruitful discussions. 
\end{acknowledgements}

\section{Preliminaries}\label{section:preliminaries}

First, we discuss the theoretical background of this paper. We give the basic definitions and remarks on the notation used in this paper. 

\begin{notation}
    In this paper, we will generally use overlined letters (e.\,g. $\ol{x}$) to denote tuples.
\end{notation}

As the name suggests, algebraic circuit classes make use of concepts originating from abstract algebra. We will first give an overview of some fundamental concepts used in this paper. For further information on the topic, the reader may refer to the books by Aluffi~\cite{aluffi2009algebra} or Lang~\cite{Lang_2002}.

\begin{definition}[Ring] \label{definition:ring}
    A \emph{ring with unit} $(R, +, \times)$ is a tuple consisting of a set $R$ and two binary operations $+$ and $\times$ such that
    \begin{itemize}
        \item $1 \in R$
        \item $(R,+)$ is an abelian group, i.\,e.
        \begin{itemize}
 	        \item For every $a,b,c \in R$, we have $(a+b)+c = a+(b+c)$
 	        \item For every $a,b \in R$, we have $a+b=b+a$
 	        \item There is a $0 \in R$, we have $a+0=a$
 	        \item For every $a \in R $, there exists an inverse $-a \in R: a+(-a) = 0 $
        \end{itemize}
        \item $(R,\times)$ is a monoid, i.\,e.
        \begin{itemize}
 	        \item $(a\times b) \times c = a\times (b \times c)$
 	        \item There is a $1 \in R$ such that $a\times 1 = 1\times a = a$
        \end{itemize}
        \item multiplication is distributive with respect to addition:
        \begin{align*}
            a \times (b+c) &= a\times b + a \times c, \text{ for all } a,b,c \in R\\
            (a+b)\times c &= a\times c + b \times c, \text{ for all } a,b,c \in R
        \end{align*}
    \end{itemize}
    A ring is called \emph{commutative}, if $a\times b = b \times a$, for all $a, b\in R$.
\end{definition}

In the following a ring is always assumed to have a unit.

\begin{definition}[Integral domain]
    A ring $R$ with unit is called integral domain, if it is a commutative ring and if there is no element $a \in R$ which is a zero-divisor, i.\,e. there is no $a\in R\setminus \{0\}$ for which there is a $b \in R \setminus \{0\}$ with $a\times b = 0$. 
\end{definition}

Throughout the paper, whenever not otherwise specified, we use $R$ to denote an infinite integral domain.

\begin{remark}
    In particular we require that for every $r \in R$ the equations
    \begin{align*}
        r \times 0 &= 0 \\
        0 \times r &= 0
    \end{align*}
    hold.
\end{remark}

\begin{definition}[Polynomial]
    Let $R$ be a ring. A \emph{polynomial} $f(x)$ in the \emph{indeterminate} $x$ and with \emph{coefficients} in $R$ is a \emph{finite} linear combination of nonnegative exponents of $x$ with coefficients in $R$:
    \begin{align*}
        f(x) = \sum_{i\ge 0} a_ix^i = a_0 + a_1x + a_2x^2 + \cdots,
    \end{align*}
    where all $a_i$ are elements of $R$ (the \emph{coefficients}) and we require $a_i = 0$ for infinitely many $i$. Two polynomials are taken to be equal if all the coefficients are equal:
    \begin{align*}
        \sum_{i\ge 0}a_ix^i = \sum_{i\ge 0}b_ix^i \iff (\forall i \ge 0)\colon a_i = b_i.
    \end{align*}
    The set of polynomials in $x$ over $R$ is denoted by $R[x]$.
\end{definition}

\begin{definition}[Polynomial ring]\label{def:polynomial}
    Let $f(x) = \sum_{i\ge 0} a_ix^i$ and $g(x) = \sum_{i\ge 0} b_ix^i$ be two polynomials. A \emph{polynomial ring} is the set $R[x]$ together with the following two operations:
    \begin{align*}
        f(x) + g(x) \coloneqq \sum_{i\ge 0} (a_i + b_i)x^i
    \end{align*}
    and
    \begin{align*}
        f(x) \times g(x) \coloneqq \sum_{k\ge 0} \sum_{\substack{0\le i\le k,\\0\le j\le k,\\i + j = k}} a_ib_j x^{i+j}
    \end{align*}
    The unit of this ring is $1_{R[x]} = 1_R + 0x + 0x^2 +\cdots$.
\end{definition}

\begin{remark}
    A polynomial ring like in Definition~\ref{def:polynomial} is also a ring with unit. Analogous to Definition~\ref{def:polynomial}, we can a define polyomial ring $S$ with finitely many variables $x_1,\ldots,x_n$ by applying Definition~\ref{def:polynomial} inductively:
    \begin{align*}
        S \coloneqq R[x_1,\ldots,x_n] \coloneqq (R[x_1,\ldots,x_{n-1}])[x_n],
    \end{align*}
    where $R$ is a ring. Note that polynomial rings over fields are integral domains.
\end{remark}

\begin{definition}[Adjoint elements to a ring]
    When \emph{adjoining} a number $f \not\in R$ to a ring $R$ we obtain the set 
    \begin{align*}
        R[f] \coloneqq \{a + f\cdot b \mid a,b \in R \}.
    \end{align*}
    The set $R[f]$ together with the following two operations
    \begin{align*}
        (a + f\cdot b) + (c + f\cdot d)  \coloneqq (a+c) + f\cdot(b+d)
    \end{align*}
    and
    \begin{align*}
        (a + f\cdot b) \cdot (c + f\cdot d)  \coloneqq ac + f(ad+bc) + f^2(bd),
    \end{align*} where $+$ and $\cdot$ are the binary operations on the ring, is again a ring, where $1_{R[x]} = 1_R + 0f$ is the unit of this ring.
\end{definition}

\begin{example}
    Popular examples of rings include $\Z,\Q,\R$ and $\C$, as well as sets with \emph{adjoined} elements like $\Z[i], \Z[\sqrt[p]{-1}], \Z[\sqrt{p},\sqrt{q},\ldots]$, $\R[i_{k}], \Z[i_{k}], \ldots$, where $k\in \mathbb{N}$, $p \neq q$ are primes and $i_{k}$ denotes the $k$th root of $-1$, i.\,e. $i_{k}^k = -1$.
    Alternatively, we can adjoin $\ell$ algebraic independent prime root elements with $2^\ell = k$ and get an analogous construction.
\end{example}

\begin{remark}
    The denotation of $\R[i_{k}]$ resp. $\Z[i_{k}]$ are not unique but the constructions of the circuit over the underlying rings are analogous except for the placeholders for the adjoined numbers. 
    For example $\Z^4$ can denote $\Z[i_4]= \Z[i_4,i_4^2,i_4^3]$ or $\Z[\sqrt{2},\sqrt{5}]= \Z[\sqrt{2},\sqrt{5},\sqrt{10}]$ or many other rings. 
    Since we only focus on the structure of the tuples, i.\,e., the coefficients of adjoint elements, or later the underlying circuits, our short notation for $\Z^k$ for arbitrary $k>1$ is not unique and may differ, e.\,g. $\Z^2$ may denote $\Z[i]$ or may denote $\Q$. 
    In the context of the placeholder notation, the arithmetic of the specific ring must be clear, if it is important. 
\end{remark}

\begin{definition}[Field]
    A field $(F,+,\times)$ is a tuple consisting of a set $F$ and two binary operations $+$ and $\times$ such that
    \begin{itemize}
        \item $1 \in F$
        \item $0 \in F$
        \item $(F,+)$ is an abelian group (see \ref{definition:ring})
        \item $(F,\times)$ is an abelian group (see \ref{definition:ring})
        \item multiplication is distributive with respect to addition:
        \begin{align*}
            a \times (b+c) = a\times b + a \times c, \text{ for all } a,b,c \in F
        \end{align*}
    \end{itemize}
    We do not need two equations for distributivity since multiplication and addition are both commutative.
\end{definition}

\begin{remark}
    An alternative characterisation for a field is as a commutative ring where $0\neq 1$ holds and every element $a\neq 0 \in F$ is invertible.
\end{remark}

We need some \emph{ordering} $<$ on our integral domain $R$. 
In some cases, like $\R$ and $\Z$, we have some natural ordering that we want to use. 
In other cases, like $\C$, we have to construct some ordering. 
This ordering does not have to be closed under multiplication, we only have to distinguish different numbers from each other. 
So an ordering on tuples ($z_1=a_1+b_1i, z_2=a_2+b_2i \in \C: (a_1,b_1) <_\C (a_2,b_2) \iff a_1 <_\R a_2 \text{ or }a_1 = a_2  \text{ and } b_1 <_\R b_2$) is possible. \smallskip

\begin{definition}
    A \emph{strict total order} on a set $S$ is a binary relation $<$ on $S$ which is irreflexive, transitive and total.
\end{definition}

\begin{example}
    For some field $\mathbb{F}_{p^k}$ for some prime $p$ and a natural number $k$, we write the finitely many elements in a list and then use the lexicographical on this list. For example, let $R=\Z_3$. Then we define $<_{\Z_3}$ as $0<1<2$.
    
    There are multiple possibilities for such a $<$ over the field of complex numbers. Let $z = a+bi \in \C$ and let $<_1 $ be the lexicographic order on pairs $(\sqrt{a^2+b^2},a)$. Furthermore, let $<_2$ be the lexicographic order on pairs of the form $(a,b)$. Then both variants are possible since both distinguish different complex numbers.
\end{example}

\begin{definition}\label{definition:sign-from-order}
    A strict total order $<$ over a ring $R$ induces a $\sign$-function as follows:
    \[
        \sign_{(R,<)}(x) = 
        \begin{cases}
            1 & \text{if $x > 0$}\\
            0 & \text{otherwise.}
        \end{cases}\qedhere
    \]
\end{definition}

In the following, unless explicitly otherwise specified, the symbol $R$ denotes an infinite integral domain with a strict total order $<$ on $R$. Most of the rings we consider have a natural ordering. In this case, we omit the ordering symbol.

\subsection{Algebraic Circuits over \texorpdfstring{$R$}{\textit{R}}}

As this work is a generalization of the well established Boolean circuits, some background in circuit complexity is assumed. 
Standard literature which introduces this topic is the book by Vollmer~\cite{DBLP:books/daglib/0097931}. 
The generalization to \emph{algebraic circuits} over real numbers were first introduced by Cucker~\cite{DBLP:journals/jc/Cucker92}. 
In analogy to them, Barlag and Vollmer defined an unbounded fan-in version of algebraic circuits~\cite{DBLP:conf/wollic/BarlagV21}. \smallskip

\begin{definition}
    We define an \emph{algebraic circuit $C$ over an integral domain $R$ with a strict total order $<$}, or $R$-circuit for short, as a directed acyclic graph. It has gates of the following types:
    \begin{description}
        \item[Input nodes] have indegree $0$ and contain the respective input values of the circuit.
        \item[Constant nodes] have indegree $0$ and are labelled with elements of $R$
        \item[Arithmetic nodes] have an arbitrary indegree $\ge 1$, bounded by the number of nodes in the circuit and are labelled with either $+$ or $\times$.
        \item[Comparison ($<$) nodes] have indegree $2$.
        \item[Output nodes] have indegree $1$ and contain the output value after the computation of the circuit.
    \end{description}

    Nodes cannot be predecessors of the same node more than once, thus the outdegree of nodes in these algebraic circuits is bounded by the number of gates in the circuit.

    During the computation of an algebraic circuit, the arithmetic gates compute their respective functions with the values of their predecessor gates being taken as the function arguments and the comparison gates compute the characteristic function of $<$ in the same way. The values of the output gates at the end of the computation are the result of the computation.
\end{definition}

In contrast to the classical setting, where we consider words over an alphabet $\Sigma$ as inputs, and where languages are thus defined to be subsets of the Kleene closure of the alphabet (in symbols, $L\subseteq \Sigma^*$), we consider vectors of integral domain elements as input. In analogy to $\Sigma^*$, we denote for an integral domain $R$:
\begin{align*}
    R^* = \bigcup_{k \in \N_0}R^k
\end{align*}
With $\abs{x}$, we denote the length of $x$,  i.\,e., if $x \in R^k$ then $\abs{x} = k$.

\begin{remark}
    We will use the term \emph{algebraic circuit} to describe circuits in the sense of Blum, Cucker, Shub and Smale~\cite{ComplexityRealComp} rather than \emph{arithmetic circuits} as e.\,g. in~\cite{DBLP:conf/wollic/BarlagV21} to distinguish them from arithmetic circuits in the sense of Valiant, see e.\,g.~\cite{AlgebraicComplTheory}. Valiant circuits are essentially algebraic circuits without $\sign$-gates~\cite[page 350]{ComplexityRealComp}. 
\end{remark}

\begin{remark}
    In the special case $R = \Z_2$, the definition above yields exactly the Boolean circuits.
\end{remark}

\begin{definition}\label{definition:size-and-depth}
    We call the number of gates in a circuit the \emph{size} of the circuit and the longest path from an input gate to an output gate the \emph{depth} of the circuit.
\end{definition}

\begin{remark}
    Unlike the way algebraic circuits with unbounded fan-in gates were introduced previously~\cite{DBLP:conf/wollic/BarlagV21}, algebraic circuits in this context have comparison gates instead of $\sign$-gates. This stems from the fact that when dealing with real or complex numbers, we can construct a $\sign$-function from $<$ via Definition~\ref{definition:sign-from-order} and the order relation from the $\sign$-function via
    \begin{align*}
        x<y \iff \sign(\sign(y-x) \cdot (2 + \sign(x-y))) = 1.
    \end{align*}
    
    If we now consider finite integral domains, however, suddenly it becomes less clear how to construct the $<$ relation from the $\sign$ function, while the other way around still works by Definition~\ref{definition:sign-from-order}.

    Given that we define circuits and logic fragments relative to an ordering, it is natural for both to have access to this ordering. Therefore we choose to use $<$ gates rather than $\sign$ gates. So in the following, all usages of $\sign$ are implicit uses of the order relation as by Definition~\ref{definition:sign-from-order}.
    
    In the cases which are considered in other literature ($\mathbb R$ or $\mathbb F_2$), this has no complexity-theoretic impact, since in the emulation of one formalism using the other, we get a linear overhead in the size and constant overhead in the depth of the circuit.
\end{remark}

In order to define complexity classes with respect to algebraic circuits, we have to define the function calculated by such a circuit and define the term of circuit families.

\begin{definition}
    The ($n$-ary) function $f_C\colon R^n \to R^m$ computed by an $R$-circuit $C$ (with $n$ input gates and $m$ output gates) is defined by
    \begin{align*}
        f_C(x_1, \dots, x_n) = (y_1, ..., y_m),    
    \end{align*}
    where $y_1, \dots, y_m$ are the values of the output gates of $C$, when given $x_1, ..., x_n$ as its inputs.
\end{definition}

\begin{definition}
    A \emph{family of $R$-circuits $\cC = (C_n)_{n \in \N}$} is a sequence of circuits which contains one circuit for every input length $n \in \N$. The function $f_\cC\colon R^*\to R^* $ computed by a circuit family $\cC$ is defined by
    \begin{align*}
        f_\cC(x) = f_{C_{\abs{x}}}(x)
    \end{align*}
    The size (resp. depth) of a circuit family $(C_n)_{n \in \N}$ is defined as a function mapping $n \in \N$ to the size (resp. depth) of $C_n$.
\end{definition}

Analogously to the classical case, we say that a set $S\subseteq R^*$ can be \emph{decided} by a circuit family $\mathcal C$, if $\mathcal C$ can compute the characteristic function of $S$. 

\begin{definition}\label{definition:AC_R}
    The class $\ACRing[R]{i}$ is the class of sets decidable by $R$-circuit families of size $\bO(n^{\bO(1)})$ and depth $\bO((\log_2 n)^i)$.
\end{definition}

\begin{definition}\label{definition:NC_R}
    The class $\NCRing[R]{i}$ is the class of sets decidable by $R$-circuit families of size $\bO(n^{\bO(1)})$ and depth $\bO((\log_2 n)^i)$, where each arithmetic gate has an indegree bounded by $2$ (we call this \emph{bounded fan-in}).
\end{definition}




\begin{remark}
    The circuit families we have just introduced do not have any restrictions on the difficulty of constructing any individual circuit given the input length.
    If it is important to know how hard obtaining a particular circuit is, one can consider so-called uniform circuit families.
    These families require their circuits to meet restrictions on the difficulties of obtaining them. 
    For more information on uniformity, cf.~\cite{DBLP:books/daglib/0097931}.
    
    Uniformity criteria can be defined for algebraic circuit classes in a similar way.
    See e.\,g.~\cite[Section 18.5]{ComplexityRealComp}.
\end{remark}

\subsection{Structures and first-order logic over integral domains}


As we want to characterize circuit complexity classes with logical fragments, this work falls broadly under the umbrella of finite model theory, and, in particular, descriptive complexity. 
Foundational knowledge of these topics is assumed and can be found in the books by Grädel et al., Libkin and Immerman~\cite{DBLP:series/txtcs/GradelKLMSVVW07, DBLP:books/sp/Libkin04, DBLP:books/daglib/0095988}. 
Traditionally, descriptive complexity is viewed as a subdiscipline of finite model theory, since allowing infinite structures often makes problems undecidable. 
In our setting, however, we want to (carefully) reintroduce infinite structures to our reasoning. For this, we use \emph{metafinite model theory}, an extension of the finite model formalism first introduced by Grädel and Gurevich~\cite{DBLP:journals/iandc/GradelG98}. 
A short introduction to metafinite model theory is also featured in the book by Grädel et al.~\cite[page 210]{DBLP:series/txtcs/GradelKLMSVVW07}. 
The approach was taken by Grädel and Meer~\cite{DBLP:conf/stoc/GradelM95} to describe some essential complexity classes over real numbers. 
These descriptions were later extended by Cucker and Meer~\cite{DBLP:journals/jsyml/CuckerM99}, where, among other things, the \NC-hierarchy over real numbers was defined. 
For the characterization, these papers introduce a so-called first-order logic with arithmetics, which we will adapt to be used in our setting.

To make proofs easier, we use the well known trick that a signature consisting only of functions can emulate any predicates we might want to define via a characteristic function corresponding to the relation the predicate is interpreted as. We can furthermore emulate constants by $0$-ary functions.

\begin{definition}\label{definition:R-structure}
    Let $L_s$, $L_f$ be finite vocabularies which only contain function symbols. An \emph{$R$-structure of signature $\sigma = (L_s, L_f)$} is a pair $\struc{D} = (\struc{A}, F)$ where 
    \begin{enumerate}
        \item \struc{A} is a finite structure of vocabulary $L_s$ which we call the \emph{skeleton} of \struc{D} whose universe $A$ we will refer to as the \emph{universe} of \struc{D} and whose cardinality we will refer to by $\abs{A}$ 
        \item and $F$ is a finite set which contains functions of the form $X \colon A^k \rightarrow R$ for $k \in \N$ which interpret the function symbols in $L_f$.
    \end{enumerate}

    We will use \emph{$\structR(\sigma)$} to refer to the set of all $R$-structures of signature $\sigma$ and we will assume that for any signature $\sigma = (L_s, L_f)$, the symbols in $L_s$ and $L_f$ are ordered.
\end{definition}

\begin{remark}
    In this paper, we only consider \emph{ranked structures}, i.\,e., structures, in which the skeleton is ordered.
\end{remark}


\begin{definition}[First-order logic over $R$]
    The \emph{language of first-order logic over an integral domain $R$} contains for each signature $\sigma = (L_s, L_f)$ a set of formulae and terms. The terms are divided into \emph{index terms} which take values in universe of the skeleton and \emph{number terms} which take values in $R$. These terms are inductively defined as follows:
    \begin{enumerate}
        \item The set of index terms is defined as the closure of the set of variables $V$ under applications of the function symbols of $L_s$.
        \item Any element of $R$ is a number term. 
        \item For index terms $h_1, ..., h_k$ and a $k$-ary function symbol $X \in L_f$, $X(h_1, ..., h_k)$ is a number term.
        \item If $t_1$, $t_2$ are number terms, then so are $t_1 + t_2$, $t_1 \times t_2$ and $\sign(t_1)$.
    \end{enumerate}


    Atomic formulae are equalities of index terms $h_1 = h_2$ and number terms $t_1 = t_2$, inequalities of number terms $t_1 < t_2$ and expressions of the form $P(h_1, ..., h_k)$, where $P \in L_s$ is a k-ary predicate symbol and $h_1, .., h_k$ are index terms.

    The set $\FORing[]{R}$ is the smallest set which contains the closure of atomic formulae under the logical connectives $\{\land, \lor, \neg, \to, \leftrightarrow\}$ and quantification $\exists v \psi$ and $\forall v \psi$ where $v$ ranges over \struc{A}. 
\end{definition}

For better readabilty, we will use inequalities of the form $x \leq y$ and the extensions of equalities $\ol{x} = \ol{y}$ and the inequalities $\ol{x} < \ol{y}$ and $\ol{x} \leq \ol{y}$ to tuples throughout this paper. 
Note that these extensions are easily definable from $=$ and $<$ in first-order logic.

\begin{remark}
    We call any element of $R$ a number term even though some integral domains can contain elements like $\zeta$, $\aleph$ or $\circledast$ which are not numbers. Since we want to do some calculations with these elements, we call them numbers, too.
\end{remark}


Equivalence of $\FORing[]{R}$-formulae and sets defined by $\FORing[]{R}$-formulae are done in the usual way, i.\,e., a formula $\varphi$ defines a set $S$ if and only if the elements of $S$ are exactly the encodings of $R$-structures under which $\varphi$ holds and two such formulae are said to be equivalent if and only if they define the same set.



With the goal in mind to create a logic which can define sets decided by circuits with unbounded {fan-in}, we introduce new rules for building number terms: the \textit{sum} and the \textit{product rule}. 
We will also define a further rule, which we call the \textit{maximization rule}.
This one is, however, already definable in $\FORing{R}$ and we thus do not gain any expressive power by using it.
We will use it to show that we can represent characteristic functions in $\FORing{R}$. 

\begin{definition}[Sum, product and maximization rule]\label{definition:sum_prod_max_rule}
    Let $t$ be a number term in which the variables $\ol{x}$ and the variables $\ol{w}$ occur freely and let $A$ denote the universe of the given input structure. Then 
    \begin{equation*}
        \sumi[\ol{x}](t(\ol{x}, \ol{w}))
    \end{equation*}
    is also a number term which is interpreted as $\sum_{\ol{x} \in A^{\abs{\ol{x}}}}t(\ol{x}, \ol{w})$. 
    The number terms $\prodi[\ol{x}](t(\ol{x}, \ol{w}))$ and $\maxi[\ol{x}](t(\ol{x}, \ol{w}))$ are defined analogously.

    We call these operators \emph{aggregators} and for any formula $\varphi$ containing aggregators of the above form, the variables in $\ol{x}$ are considered bound in $\varphi$.
\end{definition}

\begin{example}
    Let $\sigma = (\{\}, f_E^2)$ be the signature of weighted graphs, i.\,e. $f_E(x, y)$ gives the weight of the edge from $x$ to $y$ or $0$ if there is none.
    Let $\struc{G}$ be a (graph) structure over $\sigma$.
    Then the following $\FORing{R} + \SUM$ sentence $\varphi$ states that there is a node in the skeleton of $\struc{G}$, for which the sum of its outgoing edges is more than double the sum of the outgoing edges of any other node.
    \[
        \varphi \coloneqq \exists x \forall y (x \neq y \to \sumi[a](f_E(x, a)) > 2 \times \sumi[b](f_E(y, b)))
    \]
\end{example}

\begin{observation}
    $\FORing{R} = \FORing{R} + \MAX$.
\end{observation}
\begin{proof}
    An occurrence of $\maxi(F(i))$ essentially assures that there exists an element $x$, such that for all elements $y$, $F(x) \geq F(y)$ and takes the value of $F(x)$.
    Clearly, this can be defined in fist order logic by a formula of the form $\forall x \exists y F(x) \geq F(y)$.
\end{proof}

Furthermore, any characteristic function of a logical formula can be described in $\FORing{R}$. The proof for this runs analogously to that of Cucker and Meer~\cite[Proposition 2]{DBLP:journals/jsyml/CuckerM99}, since this proof does not make any use of special properties of the reals.

\section{\texorpdfstring{$\ACRing[R]{0} = \FORing[]{R}$}{AC\textasciicircum{0}\_{R} = FO\_{R}} }\label{section:AC^0_R-equals-FO_R}

In this section, we present a proof that $\FORing[]{R}$ captures the class $\ACRing[R]{0}$. 
The proof idea is similar to the proof of the established result by Immerman which characterizes $\ACRing[]{0}$ via $\FO$. 

When using logics to check, whether a given $R$-tuple would be accepted by a $R$-circuit, one needs to think about how this $R$-tuple would be interpreted as an $R$-structure $\struc{A}$. 
This can be done by interpreting it as the set of the circuit's input gates along with a single unary function $f_\mathrm{element} \colon A \to R$ mapping the $i$th input gate to the $i$th input of the circuit.
We call this kind of structure a \emph{circuit structure}.

In the following, we would like to extend $\FORing[]{R}$ by additional functions and relations that are not given in the input structure. 
To that end, we make a small addition to Definition~\ref{definition:R-structure} where we defined $R$-structures. 
Whenever we talk about $R$-structures over a signature $(L_s, L_f)$, we now also consider structures over signatures of the form $(L_s, L_f, L_a)$.
The additional (also ordered) vocabulary $L_a$ does not have any effect on the $R$-structure, but it contains function symbols, which can be used in a logical formula with this signature. 
This means that any $R$-structure of signature $(L_s, L_f)$ is also an $R$-structure of signature $(L_s, L_f, L_a)$  for any vocabulary $L_a$. 
The symbols in $L_a$ stand for functions that we will use to extend the expressive power of $\FORing[]{R}$ to capture various complexity classes.

\begin{definition}
    Let $F$ be a set of finite  functions. We will write $\FORing[F]{R}$ to denote the class of sets that can be defined by $\FORing[]{R}$-sentences which can make use of the functions in $F$ in addition to what they are given in their structure.
\end{definition}

Formally, this means that $\FORing[F]{R}$ describes exactly those sets $S \subseteq R^*$ for which there exists a $\FORing[]{R}$-sentence $\varphi$ over a signature ${\sigma = (L_s, L_f, L_a)}$ such that for each length $n$, there is an interpretation $I_n$ interpreting the symbols in $L_a$ as elements of $F$ such that for all $R^*$-tuples $s$ of length $n$ it holds that $s \in S$ if and only if $s$ encodes an $R$-structure over $(L_s, L_f, L_a)$ which models $\varphi$ when using $I_n$.

\begin{definition}
    We write \Arb{R} to denote the set of all functions $f \colon R^k \to R$, where $k \in \N$.
\end{definition}

\begin{theorem}\label{theorem:FOArb}
    Let R be an infinite integral domain. Then $\ACRing[R]{0} = \FORing[\Arb{R}]{R} + \SUM + \PROD$.
    
\end{theorem}

\begin{proof}
    The proof for this theorem works similarly to the construction for $\FORing[\Arb{\R}]{\R} + \SUMR + \PRODR  = \ACRing[\R]{0}$~\cite{DBLP:conf/wollic/BarlagV21}, since this construction does not make use of any special properties of the real numbers. 
    
    
    
    The basic idea for this proof is that we first show that for any $\FORing[\Arb{R}]{R} + \SUM + \PROD$-sentence $\varphi$, we can construct a circuit familiy which accepts its input if and only if the input encodes an $R$-structure that satisfies $\varphi$.
    This is basically done by mimicing the behaviour of the logical operators of $\varphi$ using the available gate types and evaluating the formula level by level. 
    A universal quantifier as in $\forall x \varphi(x)$, for instance, is implemented by using a $\sign$-gate (obtained from $<$ as per Definition~\ref{definition:sign-from-order}) on top of a multiplication gate, which has the circuits for $\varphi(a)$ for each $a$ in the skeleton of the encoded structure as its predecessors.
    
    To translate the semantics of $\FORing[]{R}$ into a circuit, we can mostly use the same translations as in the proof for the real case, as for the most part only properties of integral domains are used. 
    We need ring properties for most of these translations and in particular for universal quantifiers, we also need commutativity of multiplication and no zero dividers, hence we require integral domains.
    
    We do need to change the translation for existential quantifiers and $\lor$, however.
    In the real case, the circuit for $\exists x \varphi(x)$ is constructed similarly to the universal case with a $\sign$-gate followed by an addition gate with the circuits for $\varphi(a)$ for each $a$ in the skeleton as its predecessors.
    Since there are infinite integral domains with characteristic greater than $0$, i.\,e., where adding a positive amount of $1$ elements can yield $0$, this would not always produce the desired result.
    However, we can overcome this easily by translating $\exists x \varphi(x)$ to $\neg \forall x \neg \varphi(x)$ and $x \lor y$ to $\neg (\neg x \land \neg y)$, as negation, universal quantifiers and $\land$ are constructible with integral domain properties.
    
    For the converse inclusion, a number term $\val_d(g)$ is created when given a circuit family $(C_n)_{n \in \N}$, such that for each $n$, $\val_d(g)$ evaluates to the value of gate $g$ if $g$ is on the $d$th level of the respective circuit $C_n$.
    For this construction, with integral domain properties no changes need to be made to the formula in the real setting.
\end{proof}

\section{Algebraic Circuits and Guarded Functional Recursion}\label{section:gfr}



In this section, we generalize the logical characterizations of $\NC^1$ and $\AC^1$ by Durand, Haak and Vollmer~\cite{10.1145/3209108.3209179} to the respective complexity classes over integral domains $\NCRing[R]{1}$ and $\ACRing[R]{1}$. We furthermore extend this method to capture the entire $\NCRing[R]{}$ and $\ACRing[R]{}$ hierarchies.


In their paper, Durand, Haak and Vollmer use what they call \emph{guarded predicative recursion} to extend first-order logic in order to capture the logarithmic depth circuit complexity classes $\NC^1$ and $\AC^1$. This essentially amounts to a recursion operator, which halves a tuple of variables (in their numerical interpretation) which is handed down in each recursion step. This ensures that the total recursion depth is at most logarithmic. The halving of the variable tuple is performed by using the fact that addition is expressible in $\FO$ if $\BIT$ and $<$ are expressible~\cite{DBLP:books/daglib/0095988} in the following way
\[
    x \leq \frac{y}{2} \iff x + x \leq y.
\]
Note that we do not define the formula of the halving to be equality, since this is not possible for odd numbers. However, this is not an issue since we only want to bound the worst case recursion depth. In order to capture classes of polylogarithmic depth, we would like to find a suitable factor to replace $\frac{1}{2}$ with, so that the recursion process has polylogarithmic depth as well. As it turns out, for any $i \in \N$, we can assure a recursion depth of $\bO((\log_2 n)^i)$ by using the factor $2^{-\frac{\log_2 n}{(\log_2 n)^i}}$.

\begin{observation}
    Any number $n \in \N$ can be multiplied by the factor $2^{-\frac{\log_2 n}{(\log_2 n)^i}}$ exactly $(\log_2 n)^i$ times, before reaching $1$.
\end{observation}

\begin{proof}
    $n \cdot \left( 2^{-\frac{\log_2 n}{(\log_2 n)^i}} \right)^{(\log_2 n)^i} = n \cdot 2^{-\frac{\log_2 n}{(\log_2 n)^i} \cdot (\log_2 n)^i} = n \cdot 2^{- \log_2 n} = n \cdot \frac{1}{n} = 1$ 
\end{proof}

A more general version of this observation can be found in Lemma~\ref{lemma:bound-circuit-depth-by-factor-for-layer} in the appendix. Unfortunately, while it is simple to divide by $2$ when the \BIT predicate is available, it is not at all clear if multiplying by a factor such as $2^{-\frac{\log_2 n}{(\log_2 n)^i}}$ can be done in first-order logic.

We can, however, make use of the ability to divide by $2$ in order to achieve polylogarithmic recursion depth, by instead of dividing a number, essentially dividing the digits of a base $n$ number individually and carrying over once $0$ is reached.

Let us take for example the base $5$ number $444$.
The previously mentioned process is illustrated in the table in Figure~\ref{figure:tuple_example}.
The table is supposed to be read from top to bottom and then from left to right. 

We divide the digits of $444$ from the least to most the significant digit until $0$ is reached. So, the first step is dividing the rightmost digit of $444$ by $2$, getting from $444$ to $442$. After two more divisions of that kind, we reach $0$ and in the subsequent step we reset the rightmost digit and divide the second digit once. This works in a similar way to counting down, where instead of taking away $1$ in each step, we divide by $2$ and carry over in an analogous way. Notably, reaching $000$ from $444$ takes $63 = (\lfloor \log_2 5 \rfloor + 2)^3 - 1$ steps. This is no coincidence: It can easily be shown that for any base $n$ number of $i$ digits, this sort of process reaches the smallest possible element after less than $(\lfloor \log_2 n-1 \rfloor + 2)^i-1$ steps.

\begin{figure}
\begin{center}
   \begin{tabular}{CCCC}
        444 & 244 & 144 & 044\\
        442 & 242 & 142 & 042\\
        441 & 241 & 141 & 041\\
        440 & 240 & 140 & 040\\
        424 & 224 & 124 & 024\\
        422 & 222 & 122 & 022\\
        421 & 221 & 121 & 021\\
        420 & 220 & 120 & 020\\
        414 & 214 & 114 & 014\\
        412 & 212 & 112 & 012\\
        411 & 211 & 111 & 011\\
        410 & 210 & 110 & 010\\
        404 & 204 & 104 & 004\\
        402 & 202 & 102 & 002\\
        401 & 201 & 101 & 001\\
        400 & 200 & 100 & 000\\
    \end{tabular} 
\end{center}
\caption{
Illustration of digit-wise division until $0$ of the base $5$ number $444$ which takes $63 = (\lfloor \log_2 5 \rfloor + 2)^3 - 1$ steps.
}\label{figure:tuple_example}
\end{figure}

In order to perform divisions by $2$, as stated before, we need to be able to express the \BIT predicate in our logic. 
In the classical case with natural numbers, \BIT is defined as follows, where we assume the most significant bit of $j$'s binary representation to be the bit at index $1$:

\begin{definition}\label{definition:bitN}
    Let the relation $\BIT^2 \subseteq \N \times \N$ be defined as follows:
    \[
        \BIT \coloneqq \{(i, j) \mid \text{ the $i$th bit of the binary representation of $j$ is 1, } i, j \in \N\}.\qedhere
    \]
\end{definition}

Since we wish to logically characterize languages decided by circuit families, it is useful to briefly talk about representation of numbers in descriptive complexity. In descriptive complexity, tuples of elements from a finite, ordered domain $A$ are often associated with numbers. This is frequently done by interpreting the tuple as a base $\abs{A}$ number with each element of the tuple representing its position in the ordering of $A$.

\begin{example}\label{example:tuples_as_numbers}
    For example, let $D = \{a, b, c, d\}$ be ordered such that $a < b < c < d$. Then the tuple $(b,d,c,a) = (1,3,2,0)$ would be interpreted as the base $4$ number $1320$, which would correspond to $60$ in decimal.
\end{example}

Whenever we talk about the \emph{numerical interpretation} of a tuple, we refer to this sort of interpretation. Given that we are in this setting, we adjust the definition of \BIT to suit our purposes as follows, where again, the most significant bit of $j$'s binary representation is the bit at index $1$:

\begin{definition}\label{definition:bitA}
    For any ranked structure with universe $D$, let the relation $\BIT^2 \subseteq D^* \times D^*$ be defined as follows:
    \begin{multline*}
      \BIT \coloneqq \{(i, j) \mid \text{ when $i$ and $j$ are taken as their numerical interpretations, the } \\
      \text{$i$th bit of the binary representation of $j$ is 1, } i, j \in D^*\} \qedhere
    \end{multline*}
\end{definition}

With the \BIT predicate and an order relation, we are now able to express division by $2$ in first-order logic. 
We use the fact that whenever \BIT and an order are available, we can express multiplication and addition of numerical interpretations of tuples~\cite{DBLP:books/daglib/0095988}.
This result was shown for plain first-order logic, and since our two-sorted first-order logic $\FORing{R}$ is equivalent to first-order logic if the secondary component is ignored, we can apply it here as well.

We express division by $2$ as follows:
\[
    \ol{x} \leq \ol{y} / 2 \iff \exists \ol{z}~ \ol{x} + \ol{x} = \ol{z} \land \ol{z} \leq \ol{y}
\]

Note again, that since the numerical interpretations of tuples are natural numbers, expressing $\ol{x} \leq \ol{y} / 2$ is really as good as we can do, since $\ol{x} = \ol{y} / 2$ would not work for odd numbers.


Next, we will turn to defining the recursion operator which we have alluded to in the beginning of this section. First, we need a little bit of additional notation, i.\,e., \emph{relativized aggregators}. A relativization of an aggregator is a formula restricting which elements are considered for the aggregator.

\begin{notation}\label{notation:relativized_aggregators}
    For a number term $t$ and a $\FORing[]{R}$ formula $\varphi$ we write
    \[
        \maxi[\ol{x}].(\varphi(\ol{x}))t(\ol{x})
    \]
    as a shorthand for 
    \[
        \maxi[\ol{x}](\chi[\varphi(\ol{x})] \times t(\ol{x})).
    \]
    Analogously we write
    \[
        \sumi[\ol{x}].(\varphi(\ol{x}))t(\ol{x})
    \]
    for 
    \[
        \sumi[\ol{x}](\chi[\varphi(\ol{x})] \times t(\ol{x}))
    \]
    and
    \[
        \prodi[\ol{x}].(\varphi(\ol{x}))t(\ol{x})
    \]
    as a shorthand for 
    \[
        \prodi[\ol{x}](\chi[\varphi(\ol{x})] \times t(\ol{x}) + \chi[\neg \varphi(\ol{x})] \times 1).
    \]
\end{notation}

We now define the \GFRi operator and logics extended by \GFRi of the form $\mathcal{L}+\GFR$.
The idea is to mimic the behavior demonstrated in Figure~\ref{figure:tuple_example}.

\begin{definition}[$\GFR^i$] \label{definition:gfr}
    
    Let $F$ be a set of functions such that \BIT is definable in $\FORing{R}[F]$ and let $\mathcal{L}$ be $\FORing{R}[F]$ or a logic obtained by extending $\FORing{R}[F]$ with some construction rules (such as the sum or the product rule as per Definition~\ref{definition:sum_prod_max_rule}).
    
    For $i \geq 0$, the set of $\mathcal{L}+\GFRi$-formulae over $\sigma$ is the set of formulae by the grammar for $\mathcal{L}$ over $\sigma$ extended by the rule
    \begin{equation*}
        \varphi \coloneqq{} [f(\ol{x},\ol{y_1}, \dots, \ol{y_i}, \ol{y_{i+1}}) \equiv t(\ol{x},\ol{y_1}, \dots, \ol{y_i}, \ol{y_{i+1}},f)] \psi(f),
    \end{equation*}
    where $\psi$ is a $\mathcal{L}$ formula, $f$ is a function symbol and $t$ is a $\mathcal{L}$ number term with free variables $\ol{x}, \ol{y_1}, \dots, \ol{y_i}, \ol{y_{i+1}}$ such that all $\ol{y_j}$ for $1 \leq j \leq i$ contain the same (positive) number of variables and each (sub-)number term involving the symbol $f$ in $t$
    \begin{enumerate}
        \item is of the form $f(\ol{x}, \ol{z_1}, \dots, \ol{z_i}, \ol{z_{i+1}})$, where $\ol{z_1}, \dots, \ol{z_i}, \ol{z_{i+1}}$ are in the scope of a guarded aggregation  \label{item:guarded_aggregation}
        \[
        A_{\ol{z_1}, \dots, \ol{z_i}, \ol{z_{i+1}}}.\left(\bigvee\limits_{j = 1}^i \left( \ol{z_j} \leq \ol{y_j} / 2 \land \bigwedge\limits_{k = 1}^{j - 1} \ol{z_k} \leq \ol{y_k} \right) \land \xi(\ol{y_1}, \dots, \ol{y_i}, \ol{y_{i+1}}, \ol{z_1}, \dots, \ol{z_i}, \ol{z_{i+1}})\right)
        \] 
        with $A \in \{\textbf{max}, \textbf{sum}, \textbf{prod}\}$, $\xi \in \mathcal{L}$ with $\xi$ not containing any relation symbols for relations given in the input structure and
        \item never occurs in the scope of any aggregation (or quantification) not guarded this way.
    \end{enumerate}
    The function symbol $f$ is considered bound in
    \begin{align*}
        \varphi \coloneqq{} [f(\ol{x},\ol{y_1}, \dots, \ol{y_i}, \ol{y_{i+1}}) \equiv t(\ol{x},\ol{y_1}, \dots, \ol{y_i}, \ol{y_{i+1}},f)] \psi(f).
    \end{align*}
    The semantics for the \GFRi operator are defined as follows: 
    Let $\varphi$ be a $\mathcal{L} + \GFRi$ formula over $\sigma$ with the single $\GFRi$ occurrence 
    \begin{align*}
        [f(\ol{x},\ol{y_1}, \dots, \ol{y_i}, \ol{y_{i+1}}) \equiv t(\ol{x},\ol{y_1}, \dots, \ol{y_i}, \ol{y_{i+1}},f)] \psi(\ol{z}, f)
    \end{align*}
    and let $\struc{D} \in \structR(\sigma)$. Then, the operator which is applied to the formula $\psi$, namely $[f(\ol{x},\ol{y_1}, \dots, \ol{y_i}, \ol{y_{i+1}}) \equiv t(\ol{x},\ol{y_1}, \dots, \ol{y_i}, \ol{y_{i+1}},f)]$, defines the interpretation of $f$ in $\psi$ in the following way: For all tuples $\ol{a}, \ol{b_1}, \dots, \ol{b_i}, \ol{b_{i+1}}$ of elements of the universe $D$ of $\struc{D}$ with the same arities as $\ol{x}, \ol{y_1}, \dots, \ol{y_i}, \ol{y_{i+1}}$, respectively,
    \[
        f(\ol{a}, \ol{b_1}, \dots, \ol{b_i}, \ol{b_{i+1}}) = t(\ol{a}, \ol{b_1}, \dots, \ol{b_i}, \ol{b_{i+1}}, f).
    \]
    This means that the formula $[f(\ol{x},\ol{y_1}, \dots, \ol{y_i}, \ol{y_{i+1}}) \equiv t(\ol{x},\ol{y_1}, \dots, \ol{y_i}, \ol{y_{i+1}},f)] \psi(\ol{c}, f)$ holds for a tuple $\ol{c} \in D^{\abs{\ol{c}}}$ with the same arity as $\ol{z}$ if and only if $\struc{D} \models \psi(\ol{c}, f_I)$ where $f_I$ is the interpretation of $f$ as defined by the \GFRi operator.
    Semantics of formulae with several \GFRi operators are defined analogously.
    
    Note that the $(i+1)$th tuple does not get restricted by the guarded aggregation. 
\end{definition}

Having defined the \GFRi operator, it remains to be shown that it indeed ensures polylogarithmic recursion depth in the way that we want it to.
    
\begin{lemma} \label{lemma:gfr-yields-logarithmic-depth}
    Let $\struc{D} \in \structR(\sigma)$ and let $\varphi$ be a formula containing a \GFRi operator. Then the recursion depth of the $\GFRi$ operator is bounded by  $\bO((\log_2 n)^i)$, where $n$ is the size of the universe of $\struc{D}$.
\end{lemma}

\begin{proof}
    Let $n$ be the size of the universe of the structure under which the $\GFRi$ operator is interpreted.
    The bound for the recursion depth of the $\GFRi$ operator stems from the relativization which guards the aggregated variables. 
    Let $f(\ol{x}, \ol{z_1}, \dots, \ol{z_i}, \ol{z_{i+1}})$ be an occurrence of a $\GPR^i$ operator.
    Then the variables in $\ol{z_1}, \dots, \ol{z_i}, \ol{z_{i+1}}$ are in the scope of a guarded aggregation of the form
    \[
        A_{\ol{z_1}, \dots, \ol{z_i}, \ol{z_{i+1}}}.\left(\bigvee\limits_{j = 1}^i\left( \ol{z_j} \leq \ol{y_j} / 2 \land \bigwedge\limits_{k = 1}^{j - 1} \ol{z_k} \leq \ol{y_k} \right) \land \xi(\ol{y_1}, \dots, \ol{y_i}, \ol{y_{i+1}}, \ol{z_1}, \dots, \ol{z_i}, \ol{z_{i+1}})\right)
    \] 
    as per Definition~\ref{definition:gfr}.

    
    Let $z_j$ be the numerical interpretation of $\ol{z_j}$ for all $1 \leq j \leq i$.
    We interpret the tuple $\ol{z_1}, \dots, \ol{z_i}$ as a natural number $z$ of base $n$ where the $j$th digit of $z$ is $z_j$.
    
    First, observe that in this interpretation, the relativization of the variables used in the recursive call ensures that $z$ strictly decreases in each step. The big conjunction $\bigvee\limits_{j = 1}^i \ol{z_j} \leq \ol{y_j} / 2 \land \bigwedge\limits_{k = 1}^{j - 1} \ol{z_k} \leq \ol{y_k}$ makes sure that there is an index $j$, such that $\ol{z_j} \leq \ol{y_j} / 2$, which means that the numerical interpretation of $\ol{z_j}$ is at most half of the numerical interpretation of $\ol{y_j}$. It also ensures that all tuples with smaller indices $\ol{z_k}$ (i.\,e. the more significant tuples in the interpretation of $\ol{z_1}, \dots, \ol{z_i}$ as a base $n$ number) do not increase.
    
    Since each of the tuples $\ol{z_j}$ (in their numerical interpretation) can only get halved at most $\lfloor \log_2 n \rfloor +1$ times before reaching $0$, it takes at most $\lfloor \log_2 n \rfloor + 2$ recursion steps until a tuple other than the $i$th has been halved. In the worst case that is the $(i-1)$th tuple. This process can then be repeated at most $\lfloor \log_2 n \rfloor + 1$ times, before the tuple at the next lower index gets halved. In total, in the worst case, it takes $(\lfloor \log_2 n \rfloor + 2)^{j}$ recursion steps until the $i-j$th tuple gets halved in this process.
    
    This means, that after $(\lfloor \log_2 n \rfloor + 2)^{i} - 1$ recursion steps, the first tuple has reached $0$. Therefore, the total maximum recursion depth is $(\lfloor \log_2 n \rfloor + 2)^{i} - 1 \in \bO((\log_2 n)^i)$.
    
    This process can be thought of as counting down a base $\log_2 n$ number. The idea for it has already been visualized in Figure~\ref{figure:tuple_example}. It is also explicitly illustrated in Figure~\ref{fig_seq_812} and Figure~\ref{fig_seq_822}, for a sequence which we will define shortly in Definition~\ref{definition:sequence_d} to make use of exactly this kind of process.
\end{proof}

Since our final goal is to characterize both $\ACRing[R]{i}$ and $\NCRing[R]{i}$, we need to also define aggregators which model the properties of $\NCRing[R]{}$ circuits.
For this purpose, we introduce \emph{bounded aggregators}, i.\,e., relativized aggregators where we only consider the two elements of maximal size meeting the condition in the relativization.

\begin{definition}
We define the bounded aggregators $\sumib[\ol{x}]$ and $\prodib[\ol{x}]$.
They are used in the same way as the aggregators defined earlier in Definition~\ref{definition:sum_prod_max_rule}.
The semantics are defined as follows:
\begin{multline*}
\sumib[\ol{x}].(\varphi(\ol{x})t(\ol{x}, \ol{w}) \equiv\\
\sumi[\ol{x}].(\varphi(\ol{x}) \land \forall \ol{y}\; \forall \ol{z} ((\ol{y} \neq \ol{z} \land \ol{x} < \ol{y} \land \ol{x} < \ol{z}) \to (\neg \varphi(\ol{y}) \lor \varphi(\ol{z}))))t(\ol{x}, \ol{w})    
\end{multline*}
The bounded aggregators $\prodib[\ol{x}]$ and $\maxib[\ol{x}]$ are defined analogously.
\end{definition}


With this bounded aggregation, we can now define a slightly weaker version of the guarded functional recursion from Definition~\ref{definition:gfr}, which we call bounded guarded functional recursion \GFRbi. This allows us then to define logics of the form $\FORing{R}[F] + \GFRi$ or $\FORing{R}[F] + \GFRbi$.

\begin{definition}[\GFRbi]
A formula is in $\FORing{R}[F]+\GFRbi$ if the same conditions as in Definition~\ref{definition:gfr} are met, but instead of a guarded aggregation in (\ref{item:guarded_aggregation}), we require a bounded guarded aggregation.
\end{definition}

Our goal in the following is to characterize the $\ACRing[R]{}$ and $\NCRing[R]{}$ hierarchies using first-order logic and guarded functional recursion. For that purpose, we now define a sequence which we will later use as part of the numbers of our gates in order to encode the gates' depth into their numbers. 
The idea behind the construction of of this sequence will be that for a circuit family $(C_n)_{n \in \N}$ with depth bounded by $c \cdot (\log_2 n)^i$, each of the sequence's elements is essentially a $i$-digit base $c \cdot \lfloor\log_2 n\rfloor - 1$ number with each digit being encoded in unary and padded by zeroes. 
For readability purposes, we will refer to this encoding simply as a unary encoding. The sequence can then be seen as counting down to $0$ in that interpretation. 
This will then result in a length of $c^i \cdot \lfloor \log_2 n\rfloor^i \in \bO((\log_2 n)^i)$ for fixed $i$. We begin by introducing our conventions regarding unary encodings of numbers.

\begin{definition}
    Let $n, \ell \in \N$ such that $\ell \geq n$. 
    Then we will refer to the function $\un_\ell \colon \N \to \{0, 1\}^\ell$ defined as
    \[
    \un_\ell(n) \coloneqq 0^{\ell-n}1^n 
    \]
    as the (length $\ell$) unary encoding of $n$.
\end{definition}

\begin{definition}
    For any binary string \ol{a} of the form 
    \[
    \ol{a} = 0^{k-m}1^m
    \]
    for some $k, m$, we define the function $\uval \colon \{0, 1\}^\star \to \N$ defined as 
    \[
    \uval(\ol{a}) \coloneqq m   
    \]
    and call $\uval(\ol{a})$ the value of the unary encoding \ol{a}. 
\end{definition}

We now proceed to define the aforementioned sequence \d which we will later use to essentially encode our circuit's gates' depth into their gate numbers.

\begin{definition}\label{definition:sequence_d}
    For each $n, c, i \in \N$, we define the sequence $\d(n, c, i)$ as follows. For readability purposes we leave out the arguments $(n, c, i)$ in the definition of the sequence and only write $\d$ instead of $\d(n, c, i)$. 
    \begin{enumerate}
        \item Each element $\d_\ell$ of \d consists of $i$ tuples $\d_{\ell, j}$ ($1 \leq j \leq i$), each of which is the length $\lfloor c \cdot \log_2 n \rfloor - 1$ unary encoding of a number in $[0, \lfloor c \cdot \log_2 n \rfloor - 1]$.
        \item $\d_1 = 1 \dots 1$ (I.\,e., $\d_{1,j} = \un_{\lfloor c \cdot \log_2 n\rfloor - 1}(c \cdot \lfloor \log_2 n \rfloor - 1) = 1 \dots 1$ for all $j$ with $1 \leq j \leq i$.)
        \item $\d_{\ell + 1, i} = \begin{cases}
            \un_{\lfloor c \cdot \log_2 n\rfloor -1}(\lfloor c \cdot \log_2 n\rfloor - 1 ) &\textnormal{if } \uval(\d_{\ell, i}) = 0,\\
            \un_{\lfloor c \cdot \log_2 n\rfloor -1}(\uval(\d_{\ell, i}) - 1) & \textnormal{otherwise,}
        \end{cases}$ \\
        for all $\ell$ where $\d_\ell \neq 0 \dots 0$.
        \item $\d_{\ell + 1, j} = \begin{cases}
            \un_{\lfloor c \cdot \log_2 n\rfloor -1}(\uval(\d_{\ell, j} - 1)) & \textnormal{if } \uval(\d_{\ell, j+1}) = 0, \\
            \d_{\ell, j} & \textnormal{otherwise,}
        \end{cases}$\\
        for $j < i$.\qedhere
    \end{enumerate}
\end{definition}

Examples for the sequence \d can be seen in Figures \ref{fig_seq_812} and \ref{fig_seq_822}.

\begin{figure}
    \begin{center}
        \begin{tabular}{c|cc}
            \diagbox{$\ell$}{$i$} & $1$ & $2$ \\
            \hline 
            $1$ & $11$ & $11$\\
            $2$ & $11$ & $01$\\
            $3$ & $11$ & $00$\\
            $4$ & $01$ & $11$\\
            $5$ & $01$ & $01$\\
            $6$ & $01$ & $00$\\
            $7$ & $00$ & $11$\\
            $8$ & $00$ & $01$\\
            $9$ & $00$ & $00$
        \end{tabular}    
    \end{center}
    \caption{
    The sequence $\d(8, 1, 2)$. Since $c \cdot \lfloor\log_2 8\rfloor - 1 = 2$, each element of $\d(8, 1, 2)$ contains $i = 2$ tuples of length $2$. Each line is one element of the sequence, the columns determine the tuples in the elements. This means, that the first element here is the element $(11, 11)$, the second one is $(11, 01)$ and so on.
    }
    \label{fig_seq_812}
\end{figure}

\begin{figure}
    \begin{center}
        \begin{tabular}{c|cc}
            \diagbox{$\ell$}{$i$} & $1$ & $2$ \\
            \hline 
            $1$ & $11111$ & $11111$\\
            $2$ & $11111$ & $01111$\\
            $3$ & $11111$ & $00111$\\
            $4$ & $11111$ & $00011$\\
            $5$ & $11111$ & $00001$\\
            $6$ & $11111$ & $00000$\\
            $7$ & $01111$ & $11111$\\
            $8$ & $01111$ & $01111$\\
            $9$ & $01111$ & $00111$\\
            $10$ & $01111$ & $00011$\\
            $11$ & $01111$ & $00001$\\
            $12$ & $01111$ & $00000$\\
            $13$ & $00111$ & $11111$\\
            $14$ & $00111$ & $01111$\\
            $15$ & $00111$ & $00111$\\
            $16$ & $00111$ & $00011$\\
            $17$ & $00111$ & $00001$\\
            $18$ & $00111$ & $00000$\\
            $19$ & $00011$ & $11111$\\
            $20$ & $00011$ & $01111$\\
            $21$ & $00011$ & $00111$\\
            $22$ & $00011$ & $00011$\\
            $23$ & $00011$ & $00001$\\
            $24$ & $00011$ & $00000$\\
            $25$ & $00001$ & $11111$\\
            $26$ & $00001$ & $01111$\\
            $27$ & $00001$ & $00111$\\
            $28$ & $00001$ & $00011$\\
            $29$ & $00001$ & $00001$\\
            $30$ & $00001$ & $00000$\\
            $31$ & $00000$ & $11111$\\
            $32$ & $00000$ & $01111$\\
            $33$ & $00000$ & $00111$\\
            $34$ & $00000$ & $00011$\\
            $35$ & $00000$ & $00001$\\
            $36$ & $00000$ & $00000$\\
        \end{tabular}    
    \end{center}
    \caption{
    The sequence $\d(8, 2, 2)$. Since $c \cdot \lfloor\log_2 8\rfloor - 1 = 5$, each element of $\d(8, 2, 2)$ contains $i = 2$ tuples of length $5$. Each line is one element of the sequence, the columns determine the tuples in the elements. This means, that the first element here is the element $(11111, 11111)$, the second one is $(11111, 01111)$ and so on.
    }
    \label{fig_seq_822}
\end{figure}

\begin{remark}
    Note that substracting the value $1$ in this unary encoding can be seen as an integer division by $2$ in binary. This will become useful later when putting this into the context of guarded functional recursion with \BIT.
\end{remark}

Note that the length (i.\,e. the number of elements) of $\d(8, 1, 2)$ is \begin{align*}
    9 = 1^2 \cdot \lfloor \log_2 8 \rfloor^2 (= c^i \cdot \lfloor \log_2 n \rfloor^i).
\end{align*}
and the length of $\d(8, 2, 2)$ is 
\begin{align*}
    36 = 2^2 \cdot \lfloor \log_2 8 \rfloor^2 (= c^i \cdot \lfloor \log_2 n \rfloor^i).
\end{align*}
This is no coincidence. Next, we show that this observation holds in general.

\begin{lemma}\label{lemma:seq_d_length}
    Let $n, c, i \in \N$. 
    Then, $\d(n, c, i)$ has length $c^i \cdot \lfloor\log_2 n\rfloor^i$.
\end{lemma}
\begin{proof}
    Since for each element $e$ in $\d(n, c, i)$ the successor rule can be interpreted as subtracting $1$ from $e$ when $e$ is seen as a base $c \cdot \lfloor \log_2 n \rfloor$ number with $i$ digits, we are starting at the largest possible element in that sense (i.\,e. $1 \dots 1$, which would correspond to $(c \cdot \lfloor \log_2 n \rfloor)^i - 1$) and we are counting down to the lowest possible element (i.\,e. $0 \dots 0$, corresponding to $0$), there are exactly $(c \cdot \lfloor \log_2 n \rfloor)^i = c^i \cdot \lfloor \log_2 n\rfloor^i$ elements in $\d(n, c, i)$.
\end{proof}

The remaining problem that stands in the way of using the sequence $\d$ for the numbering of gates in descriptions for circuits is that the length of the elements in $\d$ depends on $n$ (which will be the number of input gates of our circuit). However, we can remedy this, since we essentially have access to base $n$ numbers in the description for a circuit with $n$ input gates (by virtue of interpreting circuit inputs as circuit structures).
Combining those with the \BIT predicate and now interpreting the unary encoded tuples in elements of $\d$ as binary numbers allows us to encode elements of $\d$ using a constant number of digits.

\begin{observation}
Let $n, c \in \N$ and $1 \leq \ell \leq c \cdot \lfloor \log_2 n \rfloor$.
The number $2^\ell - 1$ can be encoded by a base $n$ number of length $c$.
\end{observation}
\begin{proof}
    The largest possible number of that form is $2^{c \cdot \lfloor \log_2 n \rfloor} - 1 \leq n^c - 1$, which corresponds to `$\underbrace{{n-1} \dots {n-1}}_{\text{$c$ times}}$' in base $n$.
    Therefore, $2^{c \cdot \lfloor \log_2 n \rfloor} - 1$ can be encoded with $c$ base $n$ digits and thus also all smaller natural numbers can be encoded in this way.
\end{proof}

We can thus encode the binary valuations of tuples in elements of $\d$ as base $n$ numbers of length $c$.
Therefore, each element of $\d$ can be encoded using $i$ base $n$ numbers of length $c$ (or $i \cdot c$ base $n$ digits).

Before we proceed to using the sequence $\d$ for circuit descriptions, we need one more lemma which provides a useful property of $\ACRing[R]{i}$ resp. $\NCRing[R]{i}$ circuits. 
We would like to be able to talk about the \textit{depth} of gates, i.\,e., the distance of a gate to the input gates of the circuit. 
For this reason, we will establish the fact that for the circuit families we investigate, circuits exist where for each gate $g$, each input-$g$ path has the same length.




\begin{lemma}\label{lemma:balanced_dag}
    Let $L$ be in $\ACRing[R]{i}$ or $\NCRing[R]{i}$ via the circuit family $\cC = (C_n)_{n \in \N}$.
    Then there exists a circuit family $\cC' = (C'_n)_{n \in \N}$ deciding $L$, such that for all $n \in \N$ and each gate $g$ in $C'_n$, each path from an input gate to $g$ in $C'_n$ has the same length. 
    We call $C'_n$ a \emph{balanced DAG}.
\end{lemma}

\begin{proof} 
    Let $L$ be in $\ACRing[R]{i}$ or $\NCRing[R]{i}$ via $\cC = (C_n)_{n \in \N}$.
    For each circuit $C_n \in \cC$, we construct a circuit $C'_n$, such that $f_{C_n} = f_{C'_n}$ and for each gate $g$ in $C'_n$, all input-$g$-paths in $C'_n$ have the same length.
    We transform $C_n$ into $C'_n$ by creating paths of dummy gates to replace edges that go over more than one level of depth.
    
    Let the depth of $C_n$ be bounded by $c_1 \cdot (\log_2 n)^i$ and let its size be bounded by $c_2 \cdot n^{c_3}$.
    We proceed by structural induction over the depth $d$ of gates $g$ in $C_n$, i.\,e., the maximum length of paths from an input gate to $g$.
    
    \begin{description}
        \item[$d = 1$:] Each gate $g$ at depth $d$ is a direct successor of an input gate. Therefore, no changes need to be made, since all gates at depth $d$ only have input-$g$ paths of length $1$ and therefore have the desired property.
        \item[$d \to d + 1$:] For each gate at depth $d+1$, all predecessors are gates of depth $< d+1$ for which it holds that all paths from input gates to them are of the same length. Keep all those predecessors at depth $d$ as they are and replace the edge from predecessors of smaller depths $d' < d$ to $g$ by a path of dummy (unary addition) gates of length $d - d'$. Now all paths from input gates to $g$ have exactly length $d + 1$ and we only added at most $c_1 \cdot (\log_2 n)^i$ gates per predecessor of $g$. In total, the number of dummy gates added for $g$ is bounded by $c_2 \cdot n^{c_3} \cdot c_1 \cdot (\log_2 n)^i$.
    \end{description}
    
    The resulting circuit is $C'_n$. Since addition gates with only a singular predecessor are essentially identity gates, the value of each gate in any computation of $C'_n$ remains the same. Thus, $f_{C_n} = f_{C'_n}$.
    
    Additionally, for each gate in $C_n$, we add at most $c_1 \cdot (\log_2 n)^i \cdot c_2 \cdot n^{c_3}$ gates to arrive at $C'_n$.
    Therefore, the size of $C'_n$ is bounded by $c_1 \cdot (\log_2 n)^i \cdot c_2 \cdot n^{c_3} \cdot c_2 \cdot n^{c_3} = c_1 \cdot (\log_2 n)^i \cdot (c_2 \cdot n^{c_3})^2 \in \bO(n^{\bO(1)})$.
    The depth of $C'_n$ does not change, since we only ever add gates, when longer paths within $C_n$ exist so that they end up at the same length.
    
    In total, $C'_n$ computes the same function as $C_n$ -- and thus decides $L$ -- and has the property that for each of its gates $g$, all input-$g$-paths have the same length.
\end{proof}

As previously mentioned, whenever we are dealing with balanced DAGs, we will refer to the unambiguous length from input gates to a gate $g$ as the \emph{depth} of $g$.

Now we will turn to a lemma which will then finally enable us to use the previously defined sequence $\d$ to encode our gates' depth into their circuit numbers.
This, combined with Lemma~\ref{lemma:seq_d_length}, provides us with a way to logically ensure the polylogarithmic depth of a circuit given as a circuit structure. 

\begin{lemma}
Let $L$ be in $\ACRing[R]{i}$ or $\NCRing[R]{i}$ via the circuit family $\cC = (C_n)_{n \in\N}$.
Then there exists an $\ACRing[R]{i}$ (resp. $\NCRing[R]{i}$) circuit family $\cC' = (C'_n)_{n \in \N}$ deciding $L$ with depth bounded by $c \cdot (\log_2 n)^i$, such that the gates of each circuit of $\cC'$ are numbered as base $n$ numbers, where for each path from an input gate to the output gate, the first $c \cdot i$ digits encode elements of $\d$ in the order that they appear in the sequence.

The numbering after the first $c \cdot i$ digits can be chosen arbitrarily.
\label{lemma:normal_form_seq_d}
\end{lemma}
\begin{proof}
    Let $L \in \ACRing[R]{i}$ or $L \in \NCRing[R]{i}$ via the circuit family $\cC = (C_n)_{n \in \N}$ the depth of which is bounded by $c_1 \cdot (\log_2 n)^i$. 
    Without loss of generality, let the circuits of $\cC$ be balanced DAGs as per Lemma~\ref{lemma:balanced_dag}. 
    That means that for each circuit $C_n$ in $\cC$, for each gate $g$ in $C_n$, the length of all paths from input gates to $g$ is the same.
    Let $c \in \N$ be such that $c^i \cdot \lfloor \log_2 n\rfloor^i$ is larger than the depth of $C_n$ (which is bounded by $c_1 \cdot (\log_2 n)^i$).
    We pad each path in $C_n$ to length $c^i \cdot \lfloor \log_2 n\rfloor^i$ by replacing each edge from an input gate to a gate in $C_n$ by a path of dummy gates of length $c^i \cdot \lfloor \log_2 n\rfloor^i - \textit{depth}(C_n)$ so that the resulting circuit has exactly depth $c^i \cdot \lfloor \log_2 n\rfloor^i$.
    
    We now use any base $n$ numbering for the gates of $C_n$ and for each gate $g$ prepend the $\text{depth}(g)$th element of $\d(n, c, i)$ to the number of $g$. 
    Since we made sure that each input-output-path in $C_n$ has length exactly $c^i \cdot \lfloor \log_2 n\rfloor^i$, we can encode exactly the sequence $\d(n, c, i)$ in the numbers of each input-output-path.
    So now for each input-output-path, the first $c \cdot i$ digits of gate numbers encode the elements of $\d(n, c, i)$ in the order that they appear in the sequence.
\end{proof}

With the normal form from Lemma~\ref{lemma:normal_form_seq_d} and the previous definitions, we can now turn to a theorem characterizing $\ACRing[R]{i}$ and $\NCRing[R]{i}$ logically by tying it all together.

\begin{theorem}
\phantom{a}
    \begin{enumerate}
        \item $\ACRing[R]{i} = \FORing{R}[\Arb{R}]+\SUM+\PROD+\GFRi$
        \item $\NCRing[R]{i} = \FORing{R}[\Arb{R}]+\SUM+\PROD+{\GFRbi}$
    \end{enumerate}
\end{theorem}



\begin{proof}
    We start by showing the inclusions of the circuit classes in the respective logics and will then proceed with the converse directions.\medskip
    
    \noindent{}\emph{Step 1:} $\ACRing[R]{i} \subseteq \FORing{R}[\Arb{R}] + \SUM + \PROD + \GFRi$:
    
    \noindent{}Let $L \in \ACRing[R]{i}$ via the nonuniform circuit family $\cC = (C_n)_{n \in \N}$ and let the depth of $\cC$ be bounded by $c \cdot (\log_2 n)^i$.
    We construct a $\FORing{R}[\Arb{R}] + \SUM + \PROD + \GFRi$ sentence $\varphi$ defining $L$. 
    As the circuit input is interpreted as a circuit structure, the signature $\sigma$ of $\varphi$ contains only the single unary function symbol $f_{\mathrm{element}}$. 
    
    We define the following additional relations and functions which will essentially encode our given circuits. 
    We will have access to them because of the $\Arb{R}$ extension of our logic and we use relations here instead of functions for ease of reading, since we essentially have access to relations in functional structures if we consider the respective characteristic functions of the relations instead.
    \begin{itemize}\label{itemize:circuit_relations_and_function}
        \item $G_+(\ol{x})\iff\ol{x}$ is an addition gate.
        \item $G_\times(\ol{x})\iff\ol{x}$ is a multiplication gate.
        \item $G_<(\ol{x})\iff\ol{x}$ is a $<$-gate, the left predecessor of which is lexicographically lower than the right predecessor. 
        \item $G_\mathrm{input}(\ol{x}) \iff$ $\ol{x}$ is an input gate.
        \item $G_E(\ol{x}, \ol{y})\iff\ol{y}$ is a successor gate of $\ol{x}$.
        \item $G_\mathrm{output}(\ol{x})\iff\ol{x}$ is the output gate.
        \item $G_\mathrm{const}(\ol{x})\iff\ol{x}$ is a constant gate.
        \item $f_\mathrm{const\_val}(\ol{x}) = y \in R$ $\iff$ $y$ is the value of $\ol{x}$ if $\ol{x}$ is a constant gate and $y = 0$ otherwise. 
    \end{itemize}
    Without loss of generality let all circuits of $\cC$ be in the normal form of Lemma~\ref{lemma:normal_form_seq_d} and let the numbering of $C_n$ be such that the last digit of the number of the $j$th input gate is $j$ for $1 \leq j \leq n$.
    Recall that this means that for each gate number of a gate $g$ represented as a tuple $\ol{g}$, the first $c \cdot i$ elements of $\ol{g}$ encode the  $\textit{depth}(g)$th element of $d(n, c, i)$.
    The following sentence $\varphi$ defines $L$:
    \[
        \varphi \coloneqq [f(\ol{y}) \equiv t(\ol{y}, f)] \exists \ol{a}\; G_\mathrm{output}(\ol{a}) \land f(\ol{a}) = 1
    \]
    %
    %
    where $t$ is defined as follows (with $\ol{z_j}$ denoting the $j$th $c$-long subtuple in the $c \cdot i$ long prefix of $\ol{z}$, which, as per the normal form of Lemma~\ref{lemma:normal_form_seq_d} encodes the $j$th tuple of an element of $\d(n, c, i)$):
    \[\arraycolsep=1.4pt\def\arraystretch{2.5}
    \begin{array}{rrll}
         t(\ol{y}, f) \coloneqq & \chi[G_+(\ol{y})] & \times &\sumi[\ol{z}].\left(\bigvee\limits_{j = 1}^i\left( \ol{z_j} \leq \ol{y_j} / 2 \land \bigwedge\limits_{k = 1}^{j - 1}\ol{z_k} \leq \ol{y_k}\right) \land G_\mathrm{E}(\ol{y}, \ol{z})\right) f(\ol{z})\ +\\
         & \chi[G_\times(\ol{y})] & \times & \prodi[\ol{z}].\left(\bigvee\limits_{j = 1}^i\left( \ol{z_j} \leq \ol{y_j} / 2 \land \bigwedge\limits_{k = 1}^{j - 1}\ol{z_k} \leq \ol{y_k}\right) \land G_\mathrm{E}(\ol{y}, \ol{z})\right) f(\ol{z})\ +\\
         & \chi[G_<(\ol{y})]&\times & \maxi[\ol{z}].\left(\bigvee\limits_{j = 1}^i\left( \ol{z_j} \leq \ol{y_j} / 2 \land \bigwedge\limits_{k = 1}^{j - 1}\ol{z_k} \leq \ol{y_k}\right) \land G_\mathrm{E}(\ol{y}, \ol{z})\right) \\
         & & &\Bigg(\maxi[\ol{b}].\left(\bigvee\limits_{j = 1}^i\left( \ol{b_j} \leq \ol{y_j} / 2 \land \bigwedge\limits_{k = 1}^{j - 1}\ol{b_k} \leq \ol{y_k}\right) \land G_\mathrm{E}(\ol{y}, \ol{b}) \land \ol{b} < \ol{z})\right)\\
         & & &\left( \chi[f(\ol{b}) < f(\ol{z})]\right) \Bigg)\ +\\
         & \chi[G_\mathrm{input}(\ol{y})] & \times & {f_\mathrm{element}(y_{\abs{\ol{y}}})}\ +\\
         & \chi[G_\mathrm{const}(\ol{y})] & \times & f_\mathrm{const\_val}(\ol{y})\ +\\
         & \chi[G_\mathrm{output}(\ol{y})] & \times & \sumi[\ol{z}].\left(\bigvee\limits_{j = 1}^i\left( \ol{z_j} \leq \ol{y_j} / 2 \land \bigwedge\limits_{k = 1}^{j - 1}\ol{z_k} \leq \ol{y_k}\right) \land G_\mathrm{E}(\ol{y}, \ol{z})\right) f(\ol{z}).
    \end{array}
    \]
    
    
    
    Here, the relations $G_{g}(\ol{y})$ for $g \in \{+, \times, <, \textnormal{input}, \textnormal{const}, \textnormal{output}\}$ give information about the gate type of the gate encoded by $\ol{y}$ and are provided by the $\Arb{R}$-extension of $\FORing{R}$.
    They are interpreted as mentioned above.
    The \BIT predicate is provided in the same way.
    
    We will now prove that $\varphi$ does indeed define $L$.
    Let $\ol{a} \in R^n$ be the input to $C_n$.
    We will show that for all $\ol{g} \in R^l$, where $l$ is the encoding length of a gate in $\cC$, the value of the gate encoded by $\ol{g}$ in the computation of $C_n$ when $C_n$ is given $\ol{a}$ as the input is $f(\ol{g})$. 
    Let $g$ be the gate encoded by $\ol{g}$.
    We will argue by induction on the depth of the gate $g$, i.\,e., by the distance between $g$ and an input gate.
    
    $d = 0$: 
    If $d = 0$, then $g$ is an input gate. 
    Therefore, the only summand in $t(\ol{g}, f)$ that is not trivially $0$ is the fourth one, which is equal to $f_\mathrm{element}(g_{\abs{\ol{g}}})$ (which is the value of the $g_{\abs{\ol{g}}}$th input gate).
    
    $d \to d+1$:
    Since $d + 1 > 0$, $g$ is not an input gate.
    This means that there are the following 5 possibilities for $g$:
    \begin{enumerate}
        \item $g$ is an addition gate: 
        In that case, the only summand in $t(\ol{g}, f)$ which is not trivially $0$ is the first one. 
        All predecessors of $g$ have a number, the first $c \cdot i$ digits of which are the successor of the first $c \cdot i$ digits of $g$ in the sequence $\d(n, c, i)$ because of the normal form of Lemma~\ref{lemma:normal_form_seq_d}.
        Additionally, the relativization ensures that the gate encoded by $\ol{z}$ in the respective summand has exactly the successor in the sequence $\d(n, c, i)$ of $g$'s first $c \cdot i$ digits in its first $c \cdot i$ digits.
        This means that by the induction hypothesis
        \[
            \sumi[\ol{z}].\left(\bigvee\limits_{j = 1}^i\left( \ol{z_j} \leq \ol{y_j} / 2 \land \bigwedge\limits_{k = 1}^{j - 1}\ol{z_k} \leq \ol{y_k}\right) \land G_\mathrm{E}(\ol{y}, \ol{z})\right) f(\ol{z})
        \]
        yields exactly the sum of all the values of predecessor gates of $g$.
        \item $g$ is a multiplication gate:
        Analogously to the above case, 
        \[
            \prodi[\ol{z}].\left(\bigvee\limits_{j = 1}^i \left( \ol{z_j} \leq \ol{y_j} / 2 \land \bigwedge\limits_{k = 1}^{j - 1}\ol{z_k} \leq \ol{y_k}\right) \land G_\mathrm{E}(\ol{y}, \ol{z})\right) f(\ol{z})
        \]
        yields exactly the product of all predecessor gates of $g$.
        \item $g$ is a $<$ gate: 
        In that case, the relativization
        \[
            \maxi[\ol{z}].\left(\bigvee\limits_{j = 1}^i\left( \ol{z_j} \leq \ol{y_j} / 2 \land \bigwedge\limits_{k = 1}^{j - 1}\ol{z_k} \leq \ol{y_k}\right) \land G_\mathrm{E}(\ol{y}, \ol{z})\right)
        \]
        makes sure that $\ol{z}$ is the maximum gate number of a predecessor of $g$ and
        \[
            \Bigg(\maxi[\ol{b}].\left(\bigvee\limits_{j = 1}^i\left( \ol{b_j} \leq \ol{y_j} / 2 \land \bigwedge\limits_{k = 1}^{j - 1}\ol{b_k} \leq \ol{y_k}\right) \land G_\mathrm{E}(\ol{y}, \ol{b}) \land \ol{b} < \ol{z})\right) \left( \chi[f(\ol{b}) < f(\ol{z})]\right) \Bigg)
        \]
        makes sure that $\ol{b}$ is the gate number of the other predecessor of $g$ and that therefore $\chi[f(\ol{b}) < f(\ol{z})]$ is the value of $t(\ol{g}, f)$, which is exactly the value of $g$ in the computation of $C_n$.
        \item $g$ is a constant gate:
        Since $\varphi_{\mathrm{const\_val}}(\ol{g})$ returns exactly the value of $g$, and that value is taken by $t(\ol{g}, f)$.
        \item $g$ is the output gate: 
        In that case, $g$ only has one predecessor and thus
        \[
            \sumi[\ol{z}].\left(\bigvee\limits_{j = 1}^i\left( \ol{z_j} \leq \ol{y_j} / 2 \land \bigwedge\limits_{k = 1}^{j - 1}\ol{z_k} \leq \ol{y_k}\right) \land G_\mathrm{E}(\ol{y}, \ol{z})\right) f(\ol{z})
        \]
        ensures that $g$ takes the value of that predecessor, since there is only one element matching the relativization.
    \end{enumerate}
    Finally, by
    \[
        \varphi \coloneqq [f(\ol{y}) \equiv t(\ol{y}, f)] \exists \ol{a}\; G_\mathrm{output}(\ol{a}) \land f(\ol{a}) = 1
    \]
    we make sure that there exists an output gate which has the value $1$ at the end of the computation.\medskip
    
    
    
    \noindent{}\emph{Step 2:} $\NCRing[R]{i} \subseteq \FORing{R}[\Arb{R}] + \SUM + \PROD + \GFRbi$:
    
    \noindent{}The proof for this inclusion follows in the the same as the proof for the $\ACRing[R]{i}$-case, by just replacing all guarded aggregations in the formulae by bounded guarded aggregations.\medskip
    
    \noindent{}\emph{Step 3:} $\FORing{R}[\Arb{R}] + \SUM + \PROD + \GFRi \subseteq \ACRing[R]{i}$:
    
    \noindent{}Let $L \in \FORing{R}[\Arb{R}] + \SUM + \PROD + \GFRi$ via a formula $\varphi$ over some signature $\sigma = (L_s, L_f, L_a)$ and let there be only one occurrence of a $\GFRi$ operator in $\varphi$. 
    This proof easily extends to the general case.
    This means that $\varphi$ is of the form 
    \[
        \varphi = [f(\ol{x}, \ol{y_1}, \dots, \ol{y_i}, \ol{y_{i+1}}) \equiv t(\ol{x},\ol{y_1}, \dots, \ol{y_i}, \ol{y_{i+1}},f)] \psi(f).
    \]
    We now construct an $\ACRing[R]{i}$ circuit family $\cC = (C_n)_{n \in \N}$ deciding $L$.
    
    We first construct an $\ACRing[R]{i}$ family evaluating $\varphi$ without the occurrences of $f$ as in Theorem~\ref{theorem:FOArb}.
    This is possible, since $\varphi$ is an $\FORing{R}[\Arb{R}] + \SUM + \PROD$ formula over $(L_s, L_f, L_a \cup \{f\})$ and by Theorem~\ref{theorem:FOArb}, there is such a circuit family for $\varphi$.
    
    Next, we explain how we build an $\ACRing[R]{i}$ circuit family for the whole formula $\varphi$ from this point.
    For this, we need to construct an $\ACRing[R]{i}$ circuit family computing the value of $f(\ol{x}, \ol{y})$ for each pair $\ol{a},\ol{b}$ with $\ol{a} \in A^k$ for some $k \in \N$ and $\ol{b} \in A^i$, where the length of the tuple $\ol{b}$ is exactly the exponent of the polylogarithm bounding the circuit depth.
    Notice, that $t$ is an $\FORing{R}[\Arb{R}] + \SUM + \PROD + \GFRi$ number term and can therefore be evaluated by an $\ACRing[R]{0}$ circuit family, except for the occurrences of $f$.
    We now obtain $C_n$ by taking the $n$th circuit of all those (polynomially many) $\ACRing[R]{0}$ circuit families and for all $\ol{a}, \ol{b}$ replacing the gate labeled $f(\ol{a}, \ol{b})$ by the output gate of the circuit computing $t(\ol{a}, \ol{b}, f)$.
    
    Since all occurrences of $f(\ol{x}, \ol{y})$ are in the scope of a guarded aggregation 
    \begin{align*}
        A_{\ol{z_1}, \dots, \ol{z_i}}.(\bigvee\limits_{j = 1}^i \ol{z_j} \leq \ol{y_j} / 2 \land \bigwedge\limits_{k = 1}^{j - 1} \ol{z_k} \leq \ol{y_k} \land G_\mathrm{E}(\ol{y_1}, \dots, \ol{y_i}, \ol{y_{i+1}}, \ol{z_1}, \dots, \ol{z_i})),
    \end{align*}
    the number of steps from any $f(\ol{a}, \ol{b})$ before reaching $f(\ol{a}, \ol{0})$, terminating the recursion, is bounded by $\bO((\log_2 n)^i)$ as per Lemma~\ref{lemma:gfr-yields-logarithmic-depth}.
    
    Since each such step -- computing $f(\ol{a}, \ol{b})$, when given values of the next recursive call $f(\ol{c}, \ol{d})$  -- is done by an $\ACRing[R]{0}$ circuit and therefore has constant depth, in total, any path from the first recursive call to the termination has length in $\bO((\log_2 n)^i)$. 
    Since the starting circuit deciding $\varphi$ had constant depth, the circuit we constructed now has polylogarithmic depth in total.
    And given that we only added polynomially many subcircuits with polynomially many gates each, the whole circuit is an $\ACRing[R]{i}$ circuit deciding $\varphi$.
    
    For the general case of several \GFRi operators, we construct a circuit for each operator in the same way and connect them to the circuit evaluating $\varphi$.\medskip
    
    \noindent{}\emph{Step 4:} $\FORing{R}[\Arb{R}] + \SUM + \PROD + \GFRbi \subseteq \NCRing[R]{i}$:
    
    \noindent{}This case can be proven analogously to the case for $\ACRing[R]{i}$.
    Instead of $\ACRing[R]{0}$ families for evaluating $\varphi$ and $t$, we now need to use $\NCRing[R]{1}$ families. 
    With this, we have logarithmic depth for evaluating $t$, which would generally be a problem, since repeating this $(\log_2 n)^i$ times would yield a $\NCRing[R]{i+1}$ family. 
    However, by definition, there are no occurrences of $f$ in the scope of unbounded aggregators or quantifiers.
    For the bounded aggregators, we still construct circuits for all possible values of the aggregated variables, but we only connect the output gates of the maximum two circuits satisfying the relativization.
    We can do this, since, as $\xi$ does not contain function or relation symbols given in the structure, we can predetermine which circuits will match the relativization and therefore hardwire the connections.
    In general, $\FORing{R}[\Arb{R}] + \SUM + \PROD$ formulae without unbounded quantifiers or aggregators can be evaluated in $\NCRing[R]{0}$.
    Therefore, the gates marked $f(\ol{a}, \ol{b})$ only occur at constant depth in the subcircuits for $t(\ol{a}, \ol{b}, f)$. This means, that in total, the construction leads to a depth in $\bO((\log_2 n)^i)$.
\end{proof}

As mentioned previously, the basis of the idea for guarded functional recursion was the guarded predicative recursion used for plain first-order logic~\cite{10.1145/3209108.3209179}.
The same extension to polylogarithmic recursion depth that was showcased in this paper for $\GFR$ can be applied to $\GPR$, yielding the following results.
(A formal definition of $\GPR^i$ is presented in the appendix as Definition~\ref{definition:GPR}).

\begin{corollary}
\phantom{a}
\begin{enumerate}
    \item $\FO[\mathrm{Arb}] + \GPR^i = \AC^i$
    \item $\FO[\mathrm{Arb}] + \GPR_{\mathrm{bound}}^i = \NC^i$
\end{enumerate}
\end{corollary}

\section{Relationship between versions of \texorpdfstring{$\ACRing[R]{0}$}{AC\textasciicircum{0}\_{R}} over different integral domains}\label{section:relationship-regarding-rings}

Intuitively, a circuit deciding a problem in $\ACRing[\R]{0}$ should also be able to simulate circuits deciding problems in, for example, $\ACRing[\Z]{0}$ and $\ACRing[\Z_2]{0}$. Furthermore, we should be able to simulate an $\ACRing[\Z_3]{0}$-circuit by an $\ACRing[\Z_2]{0}$-circuit by simulating the operations defined in the integral domain $\Z_3$ by operations of tuples of $\Z_2$. To formalise this intuition, we propose the notion of \emph{$\mathfrak{C}$-simulation maps}.

\begin{definition}
    Let $f_1, f_2 \colon \N \to \N$ be two functions and let $\mathfrak{C}$ be a complexity class with $\mathfrak{C} \in \{\ts{}(f_1, f_2), \sd{}(f_1, f_2), \usd{}(f_1, f_2)\}$.
    
    A \emph{$\mathfrak{C}$-simulation map} from an integral domain $R_1$ to an integral domain $R_2$ is an injective function $f \colon R_1^* \to R_2^*$ such that the following holds:
    
    For all $k \in \N_{>0}$ and $A \in \mathfrak{C}_{R_1^{k}}$ there exists an $\ell \in \N_{>0}$ and a language $B \in \mathfrak{C}_{R_2^{\ell}}$ such that for all $\ol{x} = (x_1, x_2, \dots, x_{\abs{x}}$) with $x_i \in R_1^k$ for all $i$:
    \begin{align*}
        \ol{x} \in A \iff f(\ol{x}) = \left(f(x_1), f(x_2), \dots, f(x_{\abs{x}})\right) \in B.
    \end{align*}
    If there exists a $\mathfrak{C}$-simulation map from an integral domain $R_1$ to an integral domain $R_2$, we also write $\mathfrak{C}_{R_1} \simsubseteq \mathfrak{C}_{R_2}$. The relations $\simequal$ and $\simsubsetneq$ are defined analogously.
\end{definition}

Similar to the formalism of reductions in classical complexity theory, the relation induced by $\simsubseteq^\mathfrak{C}$ is reflexive and transitive.
Since in this work, the main focus is on the $\ACRing[R]{}$ and $\NCRing[R]{}$ hierarchies, we will proceed to restrict ourselves to these classes. 
But note that the simulation methods we show trivially extend to larger complexity classes.


The formalism essentially divides our integral domains in a three-tier hierarchy.
The first tier in this hierarchy consists of the complexity classes over finite integral domains, the second tier consists of the classes over integral domains which are simulatable by $\Z$ and the third tier consists of classes over integral domains which are simulatable by $\R$.
In this setting, a complexity class over a certain integral domain is able to simulate the complexity classes over integral domains in the same tier or below. 
To make this explicit, we show the following three equalities:

\begin{theorem}\label{theorem:nc-sim}
    The following three equations hold:
    \begin{enumerate}
        \item $\NCRing[\F_p]{i} \simequal \NCRing[\F_q]{i}$ for all prime powers $p, q$ and $i\in \mathbb{N}$.
        \item $\NCRing[\Z]{i} \simequal \NCRing[{\Z[j_{k}]}]{i}$ for all $i\in \mathbb{N}$.
        \item $\NCRing[\R]{i} \simequal \NCRing[{\R[j_{k}]}]{i}$ for all $i\in \mathbb{N}$.
    \end{enumerate}
\end{theorem}

\begin{proof}
    We split the proof in two main parts: the finite case, in which simulations of integral domains are trivial, and the infinite case, where we have to be more careful about the new operations our simulating circuits execute to uphold that the structure the circuits simulate are still integral domains.
    
    \emph{The finite case.} The simulation of finite integral domains is quite trivial. For $\NCRing[\F_p]{i} \simsubseteq \NCRing[\F_q]{i}$ where $p > q$, choose the length of tuples that the $\NCRing[\F_q]{i}$-circuit uses to be the smallest $k\in \mathbb{N}$ such that $p \le q^k$ holds. Map the $p$ elements to the first $p$ elements in $\F_q^k$ and define addition an multiplication tables of constant size which simulate the addition and multiplication in $\F_p$.
    
    \emph{The infinite case.} We will prove that $\NCRing[\Z]{i} \simequal \NCRing[{\Z[j_{k}]}]{i}$ for any fixed $k\in \mathbb{N}$. Note that the following proof does not depend on the countability of the set, so the proof for the uncountable case runs analogously. The direction $\NCRing[\Z]{i} \simsubseteq \NCRing[{\Z[j_{k}]}]{i}$ is trivial. For the other direction, note that when simulating infinite integral domains, we can no longer hard-code the addition and multiplication tables in a trivial way. Furthermore, integral domains are not closed under cartesian products, since for two integral domains $R_1, R_2$, we have for $R_1 \times R_2$ that $(1, 0)\times (0, 1) = (0, 0)$ using componentwise multiplication. We fix this by still using $k$-tuples to emulate the adjoint elements and using componentwise addition, but adapting multiplication so that it simulates multiplication of two elements with adjoint elements, where the $m$-th index of the tuple stands for the $m$-th power of the adjoint element, which we call $j$ here. Explicitly, we use the integral domain $(\Z^k, +_{\Z^k}, \times_{\Z^k})$, where $+_{\Z^k}$ is componentwise addition, and $\times_{\Z^k}$ is defined as follows:
    
    If we want to multiply two tuples of length $k$
    \begin{align*}
        \left(x_1, x_2, \dots, x_k\right) \times_{\Z^k} (y_1, y_2, \dots, y_k),
    \end{align*}
    we simulate the multiplication in the original, adjoint integral domain
    \begin{align*}
        \left(x_1 + x_2 j + \cdots + x_k j^{k - 1}\right) \times_{\Z[j_{k}]} (y_1 + y_2 j + \cdots + y_k j^{k - 1}),
    \end{align*}
    by constructing the following matrix of constant size which corresponds to expanding the multiplication:
    \begin{align*}
        A = (a_{uv}) = \begin{pmatrix}
                x_1 y_1 & x_1 y_2 j & \cdots & x_1 y_k j^{k-1}\\
                x_2 y_1 j & x_2 y_2 j^2 & \cdots & x_2 y_k j^k (= -x_2 y_k)\\
                \vdots & \ddots &  & \\
                x_k y_1 j^{k-1} & x_k y_2 j^k (= -x_k y_2) & \cdots & x_k y_k j^{2k-2}\\
            \end{pmatrix}.
    \end{align*}
    Observe that, due to the fact that $j^k = -1$, every entry $a_{uv}$ of the matrix contributes to the term at index $(u + v - 2) \mod k$ in the tuple. The entry of the resulting tuple at index $\ell$ is thus
    \[
        \sum_{\substack{1 \le u, v\le k\\(u + v - 2) \mod k = \ell}} a_{uv}.\qedhere
    \]
\end{proof}

\begin{theorem}\label{theorem:ac-sim}
    The following three equations hold:
    \begin{enumerate}
        \item $\ACRing[\F_p]{i} \simequal \ACRing[\F_q]{i}$ for all prime powers $p, q$ and $i\in \mathbb{N}$.
        \item $\ACRing[\Z]{i} \simequal \ACRing[{\Z[j_{k}]}]{i}$ for all $i\in \mathbb{N}$.
        \item $\ACRing[\R]{i} \simequal \ACRing[{\R[j_{k}]}]{i}$ for all $i\in \mathbb{N}$.
    \end{enumerate}
\end{theorem}

\begin{proof}
    We use the strategy from the proof of Theorem~\ref{theorem:nc-sim}, and show that unbounded addition and multiplication can be realized without a significant increase in the complexity.
    
    For $n$ given tuples of length $k$
    \begin{align*}
        (x_{1, 1}, x_{1, 2}, \dots, x_{1, k}), (x_{2, 1}, x_{2, 2}, \dots, x_{2, k}), \dots, (x_{n, 1}, x_{n, 2}, \dots, x_{n, k}),
    \end{align*}
    the simulation of unbounded addition by unbounded componentwise addition of tuples is straightforward. Observe that for unbounded multiplication, simply iterating the binary multiplication would result in linear depth. But similar to the method for binary multiplication, we can compute in advance which values contribute to a given index of the resulting tuple. And since the value of a tuple entry in the result can only depend on the values $x_{i, j}$, to simulate an unbounded multiplication gate, we need only an unbounded multiplication gate with a linear number of predecessors (but possibly exponential increase in the wire complexity).
\end{proof}

\begin{corollary}
    For all $i\in \mathbb{N}$, we have $\NCRing[\R]{i} \simequal \NCRing[\C]{i}$ and $\ACRing[\R]{i} \simequal \ACRing[\C]{i}$.
\end{corollary}

\begin{proof}
    Observe that $\C = \R[j_{(2)}]$.
\end{proof}

\begin{lemma}
    For all $i\in \mathbb{N}$, we have $\NCRing[\Z]{i} \simequal \NCRing[\Q]{i}$ and $\ACRing[\R]{i} \simequal \ACRing[\C]{i}$.
\end{lemma}

\begin{proof}
    For $\NCRing[\Z]{i} \simsubseteq \NCRing[\Q]{i}$ (resp. $\ACRing[\Z]{i} \simsubseteq \ACRing[\Q]{i}$), take $f(x) \coloneqq x$ and for $\NCRing[\Q]{i} \simsubseteq \NCRing[\Z]{0}$ (resp. $\ACRing[\Z]{i} \simsubseteq \ACRing[\Q]{i}$), take $f(x) \coloneqq (a, b)$, where $x = \frac{a}{b}$.
\end{proof}


\begin{lemma}
    $\NCRing[\Q]{i} \subsetneq \NCRing[\R]{i}$.
\end{lemma}

\begin{proof}
    We use the fact that $\R$ is a trancendental field extension of $\Q$, i.\,e., there is no finite set of numbers $M$ which we can adjoin to $\Q$ in order to get $\Q[M] = \R$. In our setting this mean that we can not simulate all numbers $r \in \R$ by a set of finite tuples $(a_0, \cdots, a_n)$ of numbers where $a_0,\cdots,a_n\in \Q$.
\end{proof}

\section{Conclusion}

In this paper, we introduced algebraic complexity classes with respect to algebraic circuits over integral domains. 
We showed a logical characterization for $\ACRing[R]{0}$ and further characterizations for the $\ACRing[R]{}$ and $\NCRing[R]{}$ hierarchies, using a generalization of the $\GPR$ operator of Durand, Haak and Vollmer~\cite{10.1145/3209108.3209179}.
We constructed a formalism to be able to compare the expressiveness of complexity classes with different underlying integral domains. 
We then showed that using this formalism, we obtain a hierarchy of sets of complexity classes, each set being able to ``simulate'' the complexity classes from the sets below.


For future work it would be interesting to investigate the logical characterizations made in this paper in the uniform setting. 
We know that for the real numbers, the characterization $\ACRing[\R]{0} = \FORing{\R}[\mathrm{Arb}_\R] + \SUMR + \PRODR$ holds both non uniformly and for uniformity criteria given by polynomial time computable circuits ($\mathrm{P}_\R$-uniform), logarithmic time computable circuits ($\mathrm{LT}_\R$-uniform) and first-order definable circuits ($\FORing{\R}$-uniform)~\cite{DBLP:conf/wollic/BarlagV21}.
We believe that the results we presented hold here in analogous uniform settings as well, though this would need to be further examined.

Another open direction is to find interesting problems (potentially even complete) for these new complexity classes. 
This could even provide new insights for the classical case.

A promising approach to the separation of algebraic circuit complexity classes could be an adaption of the approach taken by Cucker, who showed that the problem $\mathrm{FER}$, which essentially asks whether a point lies on a fermat curve, seperates the (logarithmic time uniform) $\NCRing[\R]{i}$-classes~\cite{DBLP:journals/jc/Cucker92}. 
The same proof could also hold for the $\NCRing[\C]{i}$-classes.

Another model deserving of the name ``algebraic circuit'' are arithmetic circuits in the sense of Valiant~\cite{Valiant1979}. 
This model is similar to the model presented here, with the exception that generally, only addition and multiplication gates are permitted.
Maybe the ideas presented in this paper can lead to further insights with regards to this model of computation as well.

In the Boolean case, in addition to the $\mathrm{AC}$ and $\mathrm{NC}$ hierarchies, one commonly investigated hierarchy is the so called $\mathrm{SAC}$ hierarchy. 
This hierarchy is defined by bounding the fan-in of only one gate type, i.\,e., either the conjunction or the disjunction gates. 
It is known that it does not make a difference in that setting which gate type is bounded.
A possible next step is to define a sensible analogue of the SAC hierarchy in the algebraic setting.
We believe that in the algebraic case, it does make a difference which gate type we bound.
This model could then possibly be useful to investigate algebraic structures where the respective operations do not adhere to the same axioms.

\newpage

\begin{appendix}

\section{Bounding circuit depth}

\begin{definition}
    Let $C$ be a circuit and let $L_k$ be the set of all gates $g\in C$ with $\textit{depth}(g) = k$. Then we call $L_k$ the \emph{$k$th layer of $C$}.
\end{definition}

\begin{lemma}\label{lemma:bound-circuit-depth-by-factor-for-layer}
    Let $\mathcal C = (C_n)_{n\in \mathbb N}$ be a circuit family, $f$ a sublinear function and $g_C(k)$ the number of gates at layer $k$ for a circuit $C\in \mathcal C$. Furthermore, let $\alpha \coloneqq{} 2^{-\frac{\log_2 n}{f(n)}}$.
    
    A circuit $C$ has depth $\bO(f(n))$ if we request that for every layer $k$, the inequality $g_C(k) \le \lfloor \alpha \cdot g_C(k-1)\rfloor$ holds.
\end{lemma}

\begin{proof}
    We want to bound the depth of the circuit (that is, the number of layers $k$) by the sublinear function $f$. For any input size $n$, we therefore set $k = f(n) < n$. Furthermore, if we require that the inequality $g_C(k) \le \lfloor \alpha \cdot g_C(k-1)\rfloor$ holds, the circuit reaches its maximum depth when $n\alpha^k\le 1$ holds, since the factor of $\alpha$ is applied $k$ times, once for each layer. Substitution yields
    \begin{align*}
        n\cdot \alpha^k\le 1 &\Longleftrightarrow n\cdot {2^{-\frac{\log_2 n}{f(n)}}}^k\le 1\\
        & \Longleftrightarrow n\cdot 2^{-\log_2 n}\le 1.\qedhere
    \end{align*}
\end{proof}

\begin{corollary}\label{corollary:factor-for-GPR_i}
    A circuit $C$ has depth $\bO((\log_2 n)^i)$ if we request that for every layer $k$, the inequality $g_C(k) \le \left\lfloor 2^{-\frac{\log_2 n}{(\log_2 n)^i}} \cdot g_C(k-1)\right\rfloor$ holds.
\end{corollary}

\section{Guarded Predicative Recursion}\label{section:GPR}



\begin{notation}
    Similarly to the relativized aggregators in Notation~\ref{notation:relativized_aggregators}, for relativized quantifiers, we write 
    \[
    (\exists x_1, \dots, x_k.\varphi)\psi
    \]
    as a shorthand for $\exists x_1 \dots \exists x_k (\varphi \land \psi)$ and
    \[
    (\forall x_1, \dots, x_k.\varphi)\psi
    \]
    as a shorthand for $\forall x_1, \dots, x_k (\varphi \to \psi)$.
\end{notation}

For a $\FO$-formula $\varphi$ and a relation variable $P$, we write $\varphi(P^+)$ if $P$ does not occur in the scope of a negation in $\varphi$.

\begin{definition}[$\GPR^i$]\label{definition:GPR}
    Let $\mathcal{R}$ be a set of relations such that \BIT is definable in $\FO[\mathcal{R}]$.
    The set of $\FO[\mathcal{R}] + \GPR^i$-formulae over $\sigma$ is the set of formulae by the grammar for $\FO[\mathcal{R}]$-formulae over $\sigma$ extended by the rule
    \begin{align*}
        \varphi \Coloneqq{} [P(\ol{x},\ol{y_1}, \dots, \ol{y_i}, \ol{y_{i+1}}) \equiv \theta(\ol{x},\ol{y_1}, \dots, \ol{y_i}, \ol{y_{i+1}},P^+)]\psi(P^+),
    \end{align*}
    where $\psi$ and $\theta$ are $\FO[\mathcal{R}]$-formulae over $\sigma$, $\ol{x}, \ol{y_1}, \dots, \ol{y_i}, \ol{y_{i+1}}$ are tuples of variables, $P$ is a relation variable and each atomic sub-formula involving $P$ in $\theta$
    \begin{enumerate}
        \item is of the form $P(\ol{x}, \ol{y_1}, \dots, \ol{y_i}, \ol{y_{i+1}})$, the $\ol{y_1}, \dots, \ol{y_i}, \ol{y_{i+1}}$ are in the scope of a guarded quantification
        \begin{align*}
            Q\ol{y_1}, \dots, \ol{y_i}, \ol{y_{i+1}}.\left(\bigvee\limits_{j = 1}^i \left( \ol{z_j} \leq \ol{y_j} / 2 \land \bigwedge\limits_{k = 1}^{j - 1} \ol{z_k} \leq \ol{y_k}\right) \land \xi(\ol{y_1}, \dots, \ol{y_i}, \ol{y_{i+1}}, \ol{z_1}, \dots, \ol{z_i}, \ol{z_{i+1}}) \right)
        \end{align*}
        with $Q \in \{\forall,\exists\}$, $\xi \in \FO[\mathcal{R}]$ with $\xi$ not containing any relation symbols for relations given in the input structure and
    
        \item never occurs in the scope of any quantification not guarded this way.
    \end{enumerate}
    
    The semantics for $\GPR^i$ are defined analogously to the semantics for $\GFRi$ in Definition~\ref{definition:gfr}.
\end{definition}

\end{appendix}


\printbibliography{}

@inproceedings{DBLP:conf/wollic/BarlagV21,
  author    = {Timon Barlag and
               Heribert Vollmer},
  editor    = {Alexandra Silva and
               Renata Wassermann and
               Ruy J. G. B. de Queiroz},
  title     = {A Logical Characterization of Constant-Depth Circuits over the Reals},
  booktitle = {Logic, Language, Information, and Computation - 27th International
               Workshop, WoLLIC 2021, Virtual Event, October 5-8, 2021, Proceedings},
  series    = {Lecture Notes in Computer Science},
  volume    = {13038},
  pages     = {16--30},
  publisher = {Springer},
  year      = {2021},
  url       = {https://doi.org/10.1007/978-3-030-88853-4_2},
  doi       = {10.1007/978-3-030-88853-4_2},
  bibsource = {dblp computer science bibliography, https://dblp.org}
}

@article{DBLP:journals/jsyml/CuckerM99,
	author = {Felipe Cucker and Klaus Meer},
	bibsource = {dblp computer science bibliography, https://dblp.org},
	xdoi = {10.2307/2586770},
	journal = {J. Symb. Log.},
	number = {1},
	pages = {363--390},
	title = {Logics Which Capture Complexity Classes Over The Reals},
	url = {https://doi.org/10.2307/2586770},
	volume = {64},
	year = {1999}}

@article{DBLP:journals/jc/Cucker92,
	author = {Felipe Cucker},
	bibsource = {dblp computer science bibliography, https://dblp.org},
	xdoi = {10.1016/0885-064X(92)90024-6},
	journal = {J. Complexity},
	number = {3},
	pages = {230--238},
	title = {P\({}_{\mbox{ℝ}}\) {$\neq$} {NC}\({}_{\mbox{ℝ}}\)},
	url = {https://doi.org/10.1016/0885-064X(92)90024-6},
	volume = {8},
	year = {1992}
}

@inproceedings{10.1145/3209108.3209179,
author = {Durand, Arnaud and Haak, Anselm and Vollmer, Heribert},
title = {Model-Theoretic Characterization of Boolean and Arithmetic Circuit Classes of Small Depth},
year = {2018},
isbn = {9781450355834},
publisher = {Association for Computing Machinery},
address = {New York, NY, USA},
url = {https://doi.org/10.1145/3209108.3209179},
doi = {10.1145/3209108.3209179},
booktitle = {Proceedings of the 33rd Annual ACM/IEEE Symposium on Logic in Computer Science},
pages = {354–363},
numpages = {10},
location = {Oxford, United Kingdom},
series = {LICS '18}
}

@book{ComplexityRealComp,
author = {Blum, Lenore and Cucker, Felipe and Shub, Michael and Smale, Steve},
title = {Complexity and Real Computation},
year = {1998},
isbn = {978-1-4612-6873-4},
publisher = {Springer New York},
url = {https://doi.org/10.1007/978-1-4612-0701-6}
}

@book{AlgebraicComplTheory,
  author    = {Peter B{\"{u}}rgisser and
               Michael Clausen and
               Mohammad Amin Shokrollahi},
  title     = {Algebraic complexity theory},
  series    = {Grundlehren der mathematischen Wissenschaften},
  volume    = {315},
  publisher = {Springer},
  year      = {1997},
  isbn      = {3-540-60582-7},
  bibsource = {dblp computer science bibliography, https://dblp.org}
}

@book{DBLP:books/daglib/0097931,
	author = {Heribert Vollmer},
	bibsource = {dblp computer science bibliography, https://dblp.org},
	xdoi = {10.1007/978-3-662-03927-4},
	isbn = {978-3-540-64310-4},
	publisher = {Springer},
	title = {Introduction to Circuit Complexity - {A} Uniform Approach},
	url = {https://doi.org/10.1007/978-3-662-03927-4},
	year = {1999}}

@book{Lang_2002,
	doi = {10.1007/978-1-4613-0041-0},
	url = {https://doi.org/10.1007/978-1-4613-0041-0},
	year = 2002,
	publisher = {Springer New York},
	author = {Serge Lang},
	title = {Algebra}
}

@book{buergisser,
    author = {Peter Bürgisser}, 
    title = {Completeness and reduction in algebraic complexity theory},
    publisher = {Bd. 7. Springer Science \& Business Media},
    isbn = { 978-3-540-66752-0},
    year = {2013}
}

@article{DBLP:journals/iandc/GradelG98,
  author    = {Erich Gr{\"{a}}del and
               Yuri Gurevich},
  title     = {Metafinite Model Theory},
  journal   = {Inf. Comput.},
  volume    = {140},
  number    = {1},
  pages     = {26--81},
  year      = {1998},
  url       = {https://doi.org/10.1006/inco.1997.2675},
  doi       = {10.1006/inco.1997.2675},
}

@inproceedings{DBLP:conf/stoc/GradelM95,
  author    = {Erich Gr{\"{a}}del and
               Klaus Meer},
  editor    = {Frank Thomson Leighton and
               Allan Borodin},
  title     = {Descriptive complexity theory over the real numbers},
  booktitle = {Proceedings of the Twenty-Seventh Annual {ACM} Symposium on Theory
               of Computing, 29 May-1 June 1995, Las Vegas, Nevada, {USA}},
  pages     = {315--324},
  publisher = {{ACM}},
  year      = {1995},
  url       = {https://doi.org/10.1145/225058.225151},
  doi       = {10.1145/225058.225151},
  bibsource = {dblp computer science bibliography, https://dblp.org}
}

@article{Valiant1979,
  title={Completeness classes in algebra},
  author={Leslie G. Valiant},
  journal={Proceedings of the eleventh annual ACM symposium on Theory of computing},
  year={1979}
}

@book{aluffi2009algebra,
  title={Algebra: Chapter 0},
  author={Aluffi, Paolo},
  isbn={9780821847817},
  lccn={2009004043},
  series={Graduate studies in mathematics},
  year={2009},
  publisher={American Mathematical Society}
}

@book{DBLP:series/txtcs/GradelKLMSVVW07,
  author    = {Erich Gr{\"{a}}del and
               Phokion G. Kolaitis and
               Leonid Libkin and
               Maarten Marx and
               Joel Spencer and
               Moshe Y. Vardi and
               Yde Venema and
               Scott Weinstein},
  title     = {Finite Model Theory and Its Applications},
  series    = {Texts in Theoretical Computer Science. An {EATCS} Series},
  publisher = {Springer},
  year      = {2007},
  url       = {https://doi.org/10.1007/3-540-68804-8},
  doi       = {10.1007/3-540-68804-8},
  isbn      = {978-3-540-00428-8}
}

@book{DBLP:books/sp/Libkin04,
  author    = {Leonid Libkin},
  title     = {Elements of Finite Model Theory},
  series    = {Texts in Theoretical Computer Science. An {EATCS} Series},
  publisher = {Springer},
  year      = {2004},
  url       = {https://doi.org/10.1007/978-3-662-07003-1},
  doi       = {10.1007/978-3-662-07003-1},
  isbn      = {3-540-21202-7}
}

@book{DBLP:books/daglib/0095988,
  author    = {Neil Immerman},
  title     = {Descriptive complexity},
  series    = {Graduate texts in computer science},
  publisher = {Springer},
  year      = {1999},
  url       = {https://doi.org/10.1007/978-1-4612-0539-5},
  doi       = {10.1007/978-1-4612-0539-5},
  isbn      = {978-1-4612-6809-3}
}

%
%

%
%
%
%
%
%
%
\end{document}